\documentclass[11pt]{article}
\usepackage[T1]{fontenc}
\usepackage{amsfonts,amssymb,amsthm,amsmath,epsfig}
\usepackage{stmaryrd}
\usepackage{enumitem}
\usepackage{multirow}

\RequirePackage{graphicx}
\usepackage{subcaption}
\graphicspath{ {./graphs/} }

\usepackage{algorithm}
\usepackage{algpseudocode}
\usepackage{pdfpages}
\usepackage{blkarray}
\usepackage{bbm}
\usepackage{float}% http://ctan.org/pkg/float
\usepackage[normalem]{ulem}

\makeatletter
\newenvironment{breakablealgorithm}
  {% \begin{breakablealgorithm}
   \begin{center}
     \refstepcounter{algorithm}% New algorithm
     \hrule height.8pt depth0pt \kern2pt% \@fs@pre for \@fs@ruled
     \renewcommand{\caption}[2][\relax]{% Make a new \caption
       {\raggedright\textbf{\ALG@name~\thealgorithm} ##2\par}%
       \ifx\relax##1\relax % #1 is \relax
         \addcontentsline{loa}{algorithm}{\protect\numberline{\thealgorithm}##2}%
       \else % #1 is not \relax
         \addcontentsline{loa}{algorithm}{\protect\numberline{\thealgorithm}##1}%
       \fi
       \kern2pt\hrule\kern2pt
     }
  }{% \end{breakablealgorithm}
     \kern2pt\hrule\relax% \@fs@post for \@fs@ruled
   \end{center}
  }
\makeatother

\newtheorem{prop}{Proposition}
\newtheorem{corollary}{Corollary}

\newcommand{\bqn}{\begin{eqnarray}}
\newcommand{\eqn}{\end{eqnarray}}
\newcommand{\bq}{\begin*{eqnarray}}
\newcommand{\eq}{\end*{eqnarray}}
\newcommand{\tdots}{\cdot\cdot\cdot}
\DeclareMathSymbol{\sminus}{\mathbin}{AMSa}{"39}

\usepackage{url}
\usepackage{hyperref}
\hypersetup{colorlinks,%
citecolor=blue,%
filecolor=blue,%
linkcolor=red,%
urlcolor=blue,%
}

\def\B{{\bf B}}

\def\I{{\bf I}}
\def\R{{\bf R}}

\def\Q{{\bf Q}}
\def\P{{\bf P}}
\def\U{{\bf U}}

\def\U{{\bf U}}

\def\calC{{\cal C}}

\def\calH{{\cal H}}
\def\calL{{\cal L}}
\def\calP{{\cal P}}
\def\calK{{\cal K}}

\def\calB{{\cal B}}
\def\calZ{{\cal Z}}

\def\bl{\boldsymbol\lambda}
\def\bL{\boldsymbol\Lambda}

\def\bpsi{\boldsymbol\psi}

\def\0{{\bf 0}}

\def\squarebox#1{\hbox to #1{\hfill\vbox to #1{\vfill}}}

\def\bse{\begin{eqnarray*}}
\def\ese{\end{eqnarray*}}
\def\be{\begin{eqnarray}}
\def\ee{\end{eqnarray}}
\def\bsq{\begin{equation*}}
\def\esq{\end{equation*}}
\def\bq{\begin{equation}}
\def\eq{\end{equation}}

\def\sumIP1{\sum_{i=1, i\in P_1}^N}

\begin{document}

\date{}
\title{Discrete Heat Kernels on Simplicial Complexes and Its Application to Functional Brain Networks}

\author{
Sixtus Dakurah\\
University of Wisconsin-Madison\\
\textit{sdakurah@wisc.edu}
}

\maketitle 

\begin{abstract}
Networks constitute fundamental organizational structures across biological systems, although conventional graph-theoretic analyses capture exclusively pairwise interactions, thereby omitting the intricate higher-order relationships that characterize network complexity. This work proposes a unified framework for heat kernel smoothing on simplicial complexes, extending classical signal processing methodologies from vertices and edges to cycles and higher-dimensional structures. Through Hodge Laplacian, a discrete heat kernel on a finite simplicial complex $\mathcal{K}$ is constructed to smooth signals on $k$-simplices via the boundary operator $\partial_k$. Computationally efficient sparse algorithms for constructing boundary operators are developed to implement linear diffusion processes on $k$-simplices. The methodology generalizes heat kernel smoothing to $k$-simplices, utilizing boundary structure to localize topological features while maintaining homological invariance. Simulation studies demonstrate qualitative signal enhancement across vertex and edge domains following diffusion processes. Application to parcellated human brain functional connectivity networks reveals that simplex-space smoothing attenuates spurious connections while amplifying coherent anatomical architectures, establishing practical significance for computational neuroscience applications.\\

\end{abstract}

\section{Introduction}

Networks embeds the interactions of elements in a physical or virtual space. Graphs are often used to represent and describe the geometric structures of network data, and in their basic form are composed of nodes or data points in which for any two pair of nodes, there exists a possible edge between them governed by a criteria, such as correlation. The nodes and edges in a graph are often endowed with a weighting function that defines edge weights. These weights depict \textit{signals} on the graph. Modeling signals on graphs falls under the common research paradigm of \textit{graph signal processing} \cite{newman2018networks,ortega2018graph}, in which classical signal processing methodologies are extended to irregular data structures such as Graphs.

\textit{Graph signals} are ubiquitous in many biological phenomena. In brain imaging, non-invasive procedures have been devised to segregate functional regions of the human cerebral cortex which are interconnected by a dense network of cortico-cortical axonal pathways \cite{hagmann2008mapping, macdonald2000automated}. 
Anatomical surfaces, extracted from magnetic resonance images (MRI) and computed tomography (CT) are usually represented as triangle meshes, a special case of graphs \cite{chung2015unified}. Such mesh data are usually noisy. In such graph data, measurements such as cortical thickness and surface curvature are obtained at mesh vertices, which form the nodes of the graph. In neuroscience, the interconnectivity of the  axonal pathways can in themselves be represented by a graph object \cite{bullmore2009complex}.

Often signal processing techniques are applied to such data. Processing graph signals takes various forms; smoothing, denoising, filtering among others. It has had quite successful applications in neuroimaging. Diffusion wavelets have been applied to characterize the localized growth patterns of mandible surfaces obtained in CT images, while graph signal processing can be applied to filter brain activity based on concepts of spectral modes derived from brain structure, and can be extended to brain imaging classification task \cite{hu2016matched,huang2018graph}.

Despite this success, signal processing on graph structures have the limitation that, it captures only pairwise interactions. Many physical and biological systems encodes more complex structures beyond pairwise interactions. Few studies in neuroimaging and biological systems have explored higher-order interactions in networks \cite{dakurah2022modelling,dakurahregistration,ganmor2011sparse,ohiorhenuan2010sparse,yu2011higher}. The majority of methods try to extend and generalize the pairwise relationships to group-level interactions. However, this generalization in itself is limiting, as there are not explicit in the representation of the true underlying structure of higher dimensional networks.

\begin{figure*}[ht]
 \centering
 \includegraphics[width =\textwidth]{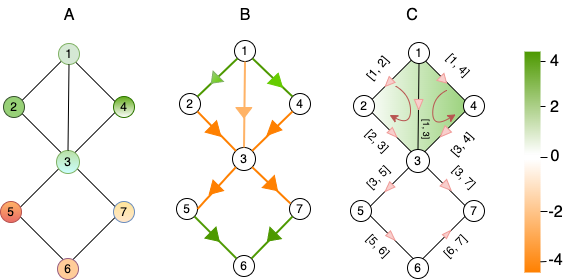}
 \caption{A schematic of signals on a simplicial complex of different orders. \textbf{(A)} Signals on 0-simplices. \textbf{(B)} Signals on 1-simplices and \textbf{(C)} Signals on 2-simplices $\{ [1, 2], [2, 3] , [1, 2]\}, \{ [1, 3], [1, 4], [3, 5] \}$. Arrows depicts the chosen orientation.}
 \label{fig:signals-on-complexes}
\end{figure*}

Simplicial complexes are algebraic objects often studied in persistent homology (PH) and topological data analysis (TDA) as the basic building block for representing complex data (see Section~\ref{sec:methods} for more details). Simplicial complex representation has been extensively applied to brain network and neuroimaging task. The encoding of the structural and functional connections of brain neurons using simplicial complexes has been studied \cite{andjelkovic2020topology,bullmore2009complex,dakurah2022modelling}, and methodologies developed for exploring their statistical properties in this representation \cite{dakurah2025maxtda,giusti2016two}. Signal processing on higher order simplices, unlike their lower order counterparts like graphs is not that straightforward. Recognizing the notion of flow in edges (Figure~\ref{fig:signals-on-complexes}), which is non-existent in signal processing over vertices, the concept of Edge-Laplacian, a special case of the Hodge-Laplacian was devised to study signals over edges \cite{schaub2018flow}. A general framework for signal processing over complexes was established, with methods developed to handle only the case of signals defined over edges , with a similar exposition on simplicial complexes as a modeling framework for higher-order network interactions, signal smoothing and denoising via Fourier transform \cite{barbarossa2020topological,schaub2018flow}.  Detail methods and analysis, however, were limited to the special case of $1$-simplices.

A popular algebraic tool for representing topological features in high dimensions is the persistent homology (PH)\cite{edelsbrunner2008persistent}. This allows for the computation of topological features of a space (represented as a simplicial complex) at different spatial resolutions. Usually this computation is preceded by imposing a function on the underlying space, which acts on the space to generate a nested sequence of simplicial complexes. The resulting nested sequence of high-order networks is referred to as a \textit{filtration}.  In medical imaging and other biological applications, PH has been applied to establish algebraic and statistical properties of various topological features in brain networks, establish statistical methods for studying electroencephalographic signals, analyzing breathing sound signals, and arterial pulse pressure waves among others \cite{emrani2014persistent,emrani2015novel}. The central theme in all these works is to develop inference procedures on topological features that persists over wide spatial scales, which are likely not to represent noise. One topological invariant that is used to distinguish topological spaces is connectedness. This invariant represents the integration of brain networks. Cycles, a second topological invariant measures the strength of this integration. Whiles brain connectivity modeled as 1-skeleton has been extensively studied in the literature\cite{chung2019statistical,park2013structural}, cycles has not enjoyed the same exploration mainly due to the fact that it's difficult to compute even for the simple case of 1-skeletons \cite{chung2019statistical,dakurahregistration}. Computing cycles in 1-skeletons is often done using brute force approach \cite{tarjan1972depth}. More recently, a scalable algorithm for computing cycles in 1-skeletons and a method for computing and identifying cycles in high-order topological spaces via the Hodge Laplacian has been developed \cite{anand2021hodge,dakurah2022modelling}. This provides an efficient algorithm for extracting cycles that retains the simplicity of graph-based approaches.

The contributions of this work are fourfold: (1) A unified heat kernel framework for signal smoothing on simplices of arbitrary dimension within finite simplicial complexes, extending beyond conventional vertex and edge processing through the Hodge Laplacian. (2) A novel cycle-preserving diffusion methodology that smooths signals along $k$-cycle representatives while maintaining homology class, enabling topological feature denoising without structural destruction. (3) Scalable algorithmic implementations leveraging sparse matrix operations that circumvent computational bottlenecks of traditional persistent homology pipelines. (4) Neurobiological validation through resting-state fMRI connectivity networks, demonstrating enhanced group-level contrasts and anatomically meaningful pattern extraction while preserving network topology.

The rest of this work is structured as follows. Section~\ref{sec:methods} presents the theoretical basis of the heat kernel smoothing method, along with algorithms for sparse computation of boundary operators, and experimental illustrations of the diffusion process. Section~\ref{sec:app} presents an application of the methodology to functional brain networks of humans, and Section~\ref{sec:disc} concludes with discussion and directions for future work.

\section{Methods}\label{sec:methods}
The mathematical and computation framework pertinent to application of simplex signal processing through Hodge theory to large scale brain networks is presented in this section.

\subsection{Graphs} A graph is an ordered set $\mathcal{X} = \left(V(\mathcal{X}), E(\mathcal{X}), \psi_{\mathcal{X}} \right)$, where $V(\mathcal{X})$ is the set of nodes, $E(\mathcal{X})$ is the set of edges and $\psi_{\mathcal{X}}$ is an incidence function connecting any pair of nodes of $V(\mathcal{X})$ to form edges $E(\mathcal{X})$. When the underlying graph is unambiguous from the context we usually denote the node set and vertex set as $V$ and $E$ respectively. The number of nodes and edges are also denoted as $v(\mathcal{X})$ (also denoted as p) and $e(\mathcal{X})$ respectively. Two nodes incident by the same edge are said to be adjacent, similarly two edges incident by the same node are also adjacent \cite{bondy1991graph}. They are two important matrices associated with graphs; the adjacency matrix $\textbf{A}(\mathcal{X})$ which is a ${v(\mathcal{X})\times v(\mathcal{X})}$ matrix with the $[a_{ij}]$ entry been the number of edges connecting node $v_i$ and $v_j$, and the incidence matrix $\textbf{B}(\mathcal{X})$ which is a $v(\mathcal{X})\times e(\mathcal{X})$ matrix with the $[b_{ij}]$ entry been the number of times node $v_i$ is incident with edge $e_j$, with loops contributing two (2) \cite{bondy1991graph,gross2005graph}. $\textbf{B}(\mathcal{X})$ is an alternate form of depicting a graph and we can derive the degree of a node by taking the appropriate column sum of $\textbf{A}(\mathcal{X})$.

A graph is complete if any pair of distinct nodes is incident with an edge. A form of graph pertinent to our application is the weighted graph $G = (V, E, w)$ where $V$ and $E$ are the usual node and edge sets and $w$ is a set of numerical quantities representing some sort of weight between two nodes. A binary weighted graph denoted $\mathcal{X}_\epsilon = (V, w_{\epsilon})$ is a modification of weighted graph that assigns the binary weight $w_{\epsilon, ij} = 1{(w_{ij}> \epsilon)}$ to nodes $v_i$ and $v_j$, for some threshold value $\epsilon$.

\subsection{Simplicial Complex and Simplicial k-chains:} A k-simplex is a k-dimensional polytope whose convex hull is made up of its k+1 nodes. A simplicial complex $\mathcal{K}$ is a set of simplices such that for any $\sigma_1, \sigma_2 \in \mathcal{K}$, $\sigma_1 \cap \sigma_2$ is a face of both simplices; and a face of any simplex $\sigma \in \mathcal{K}$ is also a simplex in $\mathcal{K}$.
For some $r_i \in \mathbb{R}$, the the formal sum $\sum_{i=1}^N r_i \sigma_i$, $\sigma_i \in \mathcal{K}$ is termed the simplicial k-chain. The set of simplicial k-chains with formal addition over $\mathbb{R}$ is known as an $\mathbb{R}$-module, and is used to construct boundary maps \cite{topaz2015topological}.

\subsubsection{Boundary Map of Simplicial k-chains:} Take any two $\mathbb{R}$-modules $\mathcal{K}_k$ and $\mathcal{K}_{k-1}$, the boundary map 

\begin{equation}
    \partial_k(\sigma): \mathcal{K}_k \longrightarrow \mathcal{K}_{k-1} 
    \label{eqn:init_boundary_map}
\end{equation}
\noindent
for each k-simplex $\sigma = (v_0, ..., v_k)$ is given as: 
\begin{equation}
    \partial_k(\sigma) = \sum_{i=1}^N(-1)^i(v_0, ..., \hat{v_i}, ..., v_k)
    \label{eqn:boundary_map}
\end{equation}
where $(v_0, ..., \hat{v_i}, ..., v_k)$ is the k-1 face of $\sigma$ obtained by deleting the $v_i$ node. The boundary map satisfies the property that $\partial_k o \partial_{k+1} = 0$.

To illustrate the connection between the simplicial k-chain and the boundary map, consider the following 3-simplices: $\sigma_i = \{v_{0i}, v_{2i}, v_{3i}, v_{4i}\}$ for $i=1,2,3$. We can construct a simplicial 3-chain by taking the sum $\sum_{i=1}^N r_i \sigma_i$ for some finite N. A sample collection of such a 3-chain can have the form $\mathcal{K}_3 =  \{r_1\sigma_1 + r_3\sigma_3, r_1\sigma_1 + r_2\sigma_2, r_2\sigma_2 + r_3\sigma_3 \}$. To show that this is an $\mathbb{R}$-module, it suffice to show that it admits addition over the ring. Clearly, if we pick any two in this collection, it will be defined over addition.

\subsubsection{Homology of a Simplicial Complex $\mathcal{K}$:} Two important components of the boundary map (\ref{eqn:boundary_map}) are it's kernel denoted $\mathcal{Z}_k = ker(\partial_k)$, and it's image denoted $\mathcal{B}_{k} = img(\partial_{k+1})$. These are subspaces of $\mathcal{K}_k$. The elements of $\mathcal{Z}_k$ are known as $k$-cycles and that of $\mathcal{B}_{k}$ are known as $k$-boundaries \cite{hatcher2002algebraic}. In essence $\mathcal{Z}_k$ is the subspace of  $\mathcal{K}_k$ consisting of $k$-chains that are also $k$-cycles and $\mathcal{B}_{k}$ is a subspace of $\mathcal{K}_k$ consisting of $k$-cycles that are also $k$-boundaries.

The set quotient $\mathcal{H}_k = \mathcal{Z}_k/\mathcal{B}_{k}$ is termed the $kth$ homology-module and its an $\mathbb{R}$-module. It can be shown that $\mathcal{B}_{k} \subseteq \mathcal{Z}_k$ \cite{hatcher2002algebraic,topaz2015topological}. Hence the difference, can intuitively be thought of as the failure of k-cycles in $\mathcal{K}$ to bound $(k+1)$-simplices, leading to the concept of "holes". A measure of this number of $k$-dimensional holes is termed the Betti number denoted $\beta_{k} = rank(\mathcal{H}_k)$.

 \subsection{Homology of Graphs:} Graphs have natural relationship with complexes. An unoriented graph can be viewed as a 1-dimensional simplicial complex, and higher dimensional simplicial complexes are referred to as hypergraphs \cite{lee2012persistent}.
 Consider the following characterization of a graph $\mathcal{X} = (V, E, w)$. Define $\mathcal{V}$ and $\mathcal{E}$ to be the node space and edge space of the graph. The formal sum over the node space
 
\begin{equation}
    \sum_{v\in \mathcal{V}} z_vv, \hspace{1cm} z_v \in\mathbb{Z}
    \label{eqn:hmg_ochains}
\end{equation}
where $z_v = 0$ for all but a finite $v\in \mathcal{V}$. Note that this sum is a combination of 0-simplices (nodes) and is a 0-chain denoted $\mathcal{K}_0\mathcal{(X)}$. Similarly, a formal sum over the directed edge space (combination of 1-simplices) will give the 1-chain denoted $\mathcal{K}_1\mathcal{(X)}$. Similar to the construction of homology groups for simplicial complexes, we can define the boundary map
 
 $$\partial_1: \mathcal{K}_1\mathcal{(X)} \longrightarrow \mathcal{K}_{0}\mathcal{(X)}$$

Since $img(\partial_2) = 0$ \cite{sunada2012topological}, we have that the first homology module $\mathcal{H}_1\mathcal{(X)} = \ker(\partial_1)$. If $\mathcal{X}$ is finite, then $\mathcal{H}_1\mathcal{(X)}$ has finite rank, given by the first Betti number $\beta_1$. Elements of $\mathcal{H}_1\mathcal{(X)}$ are termed 1-cycles.

\begin{prop}
    The number of loops (or 1-cycles) in a complete graph $\mathcal{X}$ with $p$ nodes is $\frac{1}{2}p^2 -\frac{3}{2}p + 1$.
\end{prop}
\noindent
\begin{proof}
%\footnote{\red{Notation $v$ is used in two different context. If possible try to use one notation to mean only one thing through the paper.}}
Let the complete graph $\mathcal{X}$ has the form $\mathcal{X} = (V, E, w)$. Since $\mathcal{X}$ is complete, it's also connected with 1 connected component. The $0th$ Betti number $\beta_0$ is the number of connected components of $\mathcal{X}$. From the discussion in homology of graphs, the number of 1-cycles will be given by the $1st$ Betti number $\beta_1$.
The Euler characteristic function of $\mathcal{X}$ is given by $\chi(\mathcal{X}) = |V| - |E| = \beta_0 - \beta_1$ \cite{adler2010persistent}.
\noindent
The graph $\mathcal{X}$ is complete $\implies$ $|V| = p$, $|E| = \frac{|V|(|V|-1)}{2} = \frac{p(p-1)}{2}$. Hence 

    $$B_1 = 1 - p + \frac{p(p-1)}{2} = \frac{1}{2}p^2 -\frac{3}{2}p + 1 = \mathcal{O}(p^2)$$
\end{proof}

\subsubsection{Graph Filtration} Given ordered thresholds (filtration values) ${{\epsilon}_0} < ... < {{\epsilon}_k}$, a filtration of the graph $\mathcal{X}$ is a collection of a sequence of nested binary networks:

\begin{equation}
    \mathcal{X}_{{\epsilon}_0} \supset ... \supset \mathcal{X}_{{\epsilon}_k}
    \label{eqn:binary_network}
\end{equation}

Intuitively, since $\mathcal{X}_0$ is complete, the number of connected components is therefore 1, as we increase threshold $\epsilon$, more edges are disconnected, increasing the number of connected components $(\beta_0)$, and decreasing the number of cycles $(\beta_1)$. More formally, it has been proved that in a graph $\mathcal{X}$, Betti numbers $\beta_0$ and $\beta_1$ are monotone over graph filtration on edge weights \cite{songdechakraiwut2020topological}.

\begin{prop}
    The number of loops (or 1-cycles) in a binary graph $\mathcal{X}_\epsilon$ is $$\beta_1(\mathcal{X}_{\epsilon}) = \beta_0(\mathcal{X}_{1{(w_{ij} > \epsilon)}}) - |V| + f(|V|, \epsilon)$$
\end{prop}
\noindent

\begin{proof}
The Euler characteristic of $\mathcal{X}_\epsilon$ is $$\chi(\mathcal{X}_\epsilon) = \beta_0(\mathcal{X}_{\epsilon}) - \beta_1(\mathcal{X}_{\epsilon}) = |V| - f(|V|, \epsilon)$$ \cite{adler2010persistent}.
Note that $\beta_0(\mathcal{X}_{\epsilon}) = _0(\mathcal{X}_{1{(w_{ij} > \epsilon)}})$ and $|V|$ is the number of nodes in $\mathcal{X}_\epsilon$. The function $f(|V|, \epsilon)$ is a decreasing count of the edge set over the filtration. Since the first Betti number measures the number of 1-cycles, rearranging we have:
$$\beta_1(\mathcal{X}_{\epsilon}) = \beta_0(\mathcal{X}_{1{(w_{ij} > \epsilon)}}) - |V| + f(|V|, \epsilon)$$
\end{proof}

\subsubsection{Betti Curves of Graph Filtration}
We briefly introduce the general theory of constructing filtration curves and extend to the Betti curves as a special case. Consider $\mathcal{F}_{\mathcal{X}}$ to be the filtration of a graph $\mathcal{X}$, which we previously defined as a collection of a sequence of nested binary networks. We note that the graph $\mathcal{X}$ is weighted and the binary network (\ref{eqn:binary_network}) is unweighted for some ordered thresholds ${{\epsilon}_0} < ... < {{\epsilon}_k}$.

Given the defined weights on $\mathcal{X}$, we can induce a filtration by appealing to a \textit{graph descriptor function}. A \textit{graph descriptor function} function $g$ is a mapping from $\mathcal{X}$ to some $d$-dimensional space, $g:\mathcal{X} \to \mathbb{R}^d$. That's, it evaluates certain attributes of $\mathcal{X}$ and embeds them in some $d$-dimensional space.
For each of the subgraphs in (\ref{eqn:binary_network}), we can construct a higher dimensional path
\begin{equation}
    P_{\mathcal{X}} := \operatorname*{\bigoplus}\limits_{i=0}^{k} g(\mathcal{X}_{\epsilon_i}) \in \mathbb{R}^{k+1 \times d},
\end{equation}
where ${\bigoplus}$ is the direct sum, $k+1$ is the number of thresholds and $d$ can be viewed as the measured features. Its now straightforward to observe that the count of the number of connected components is a descriptor function, and so is the count of the number of cycles. Figure~\ref{fig:betticurve} provides an example filtration process and the corresponding Betti numbers or count.
\begin{figure*}[ht]
 \centering
 \includegraphics[width =\textwidth]{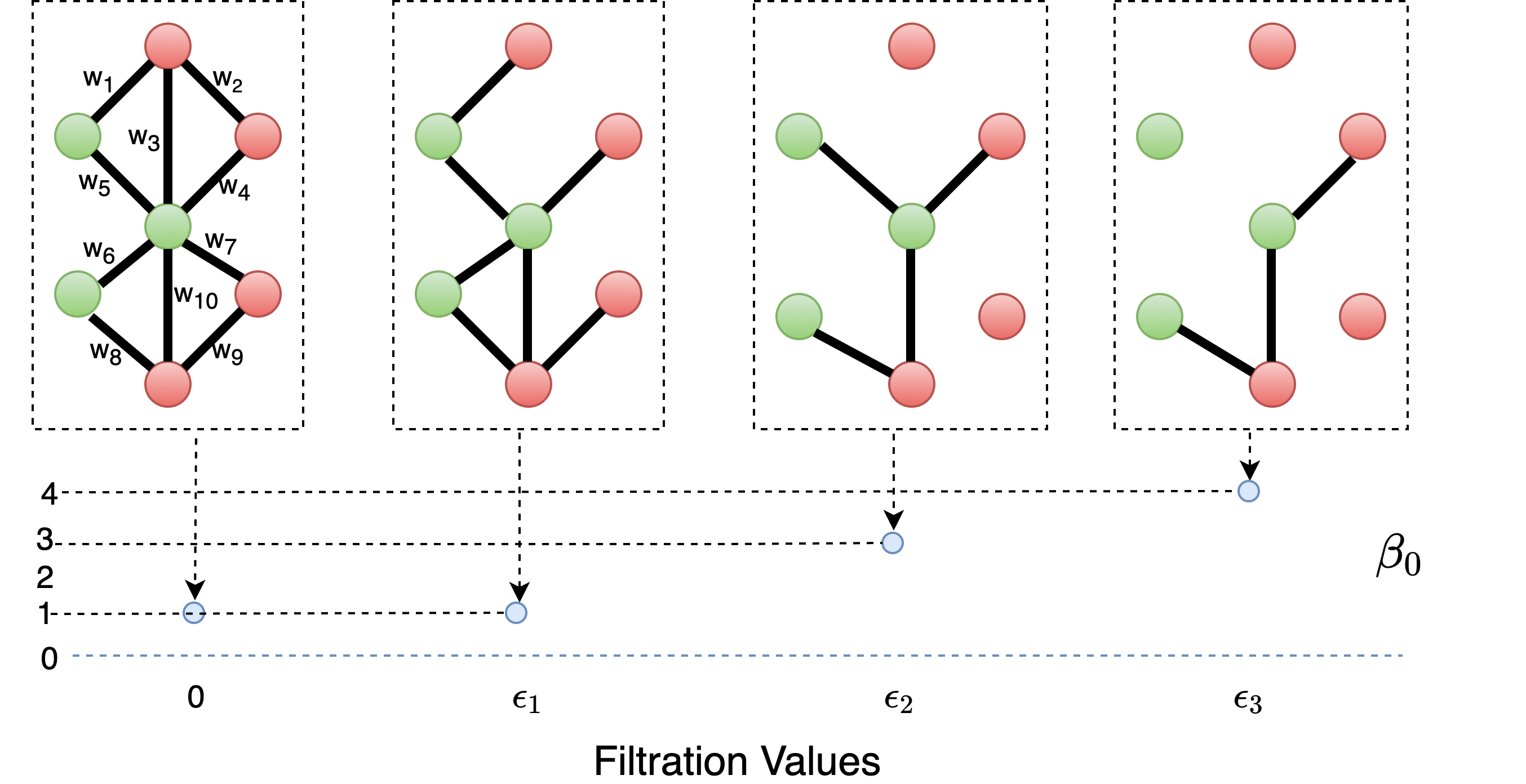}
 \caption{A schematic of a sample filtration sequence of a seven node network $\mathcal{X}$. A weighted graph (far left) is thresholded at an increasing sequence of filtration values. A graph descriptor function $g$ (count of the number of connected components) is evaluated at each filtration value. The thresholding results in a sequence of nested graphs as shown in (\ref{eqn:binary_network}). The collection of such nested sequence of graphs defines a filtration.}
 \label{fig:betticurve}
\end{figure*}
The Betti curves, which displays the Betti numbers over the network filtration value can now easily be constructed. In what follows, we'll demonstrate the curve construction process through a simulation study after we've briefly discussed random graphs.

\subsubsection{Random Graphs} Random graphs construction are mostly inspired by the Erdos-Renyi (ER) random graph \cite{bollobas1985random}. They are two theoretical approaches to generating random graphs based of the ER model. The first approach generates a random graph with a specified set of edges, say $M$, with $v$ nodes by randomly choosing an $M$-edge graph from the subset $\binom{N}{M}$ $= \frac{N!}{(N-M)!M!}$, where $N = \frac{1}{2}p(p-1)$ is the maximum number of edges in a $p$-node graph. The second approach, which is more commonly used involves assigning an edge placement probability $\pi$ to each of the $N$ possible edges. For example, assume we want to generate a random graph of order $p$ (number of nodes), we will choose each of the $N$ potential edges, run a random number generator (or equivalently flip a coin), that outputs an integer say 1 with probability $\pi$, indicating that the edge should be created, or not created otherwise. The probability space associated with this second generation process has lot of desirable properties and has been extensively studied \cite{erdos1960evolution,gilbert1959random}.

Graphs modelled after ER are often unweighted. However, complex networks will require the characterization of the heterogeneity present in edge weights. The construction of weighted random graphs has been studied as an extension of the ER model \cite{garlaschelli2009weighted}, as a generalization of the unweighted graph and also by looking at an ensemble of weighted graphs \cite{barabasi1999emergence}. In it's simplest form, once a random edge weight has been placed, a random weight (often modelled after a probability distribution) is assigned to the edge.

\subsubsection{Modularity and Random Graphs} The concept of modularity of networks or graphs is associated with the strength of the division of networks or graphs into modules. Modularity is therefore a quality index for clustering graphs \cite{newman2006modularity}.

Consider a graph $\mathcal{X} = (V, E)$ with $|V| = p$ and let $\mathcal{C} = \{C_1, ..., C_K\}$ be a partition of $V$. If $C_k \ne \emptyset $ $\forall k \in \{1, ..., K\}$, then $\mathcal{C}$ is known as a clustering of $\mathcal{X}$. It can be observed that $\mathcal{C}$ will be trivial if either $K$ is 1 or $p$. We can view each $C_k$ as an induced subgraph 
$\mathcal{X}_K = \mathcal{X}\left[C_k\right] = (C_k, E(C_k))$, where $ E(C_k) = \{ \{v_i, v_j \} \in E : v_i, v_j \in C_k \}$.
The strength of the modular structure defined by the partition $\mathcal{C}$ can be quantified via the modularity measure \cite{newman2006modularity}

\begin{equation}
    Q = \sum_{k=1}^K (e_{kk} - a_k^2 )
    \label{eqn:modulariry-measure}
\end{equation}

where $e_{kk}$ is the proportion of all edges that are within the module $C_k$ and is given as $$e_{kk} = \frac{|E(C_k)|}{|E|}$$ and $a_k$ is the fraction of all edges that have nodes in module $k$ given as $$ a_k = \frac{1}{2|E|}\sum_{\forall v_i \in C_k } d(v_i) $$ where $d(v_i)$ is the degree of node $v_i$. The objective is always to maximize $Q$ which involves maximizing the number of intramodules edges and minimizing the number of intermodule edges. The maximum value of $Q$ is 1, and values close to that indicates strong modular structure. However, empirical modularity values of graphs falls between $0.3$ and $0.7$ \cite{newman2006modularity}.

\subsubsection{Weighted Random Modular Graphs} Various methods exist for the construction of weighted random modular graphs. For example, the assignment of nodes to modules can be based on the module size distribution, and edge connection modeled after the Havel-Hakimi algorithm \cite{newman2006modularity}. Others also propose modeling edge weights based of Gaussian mixture models (GMM) or the use of GMM coupled with variational inference \cite{garlaschelli2009weighted}.
All these methods share a common theme of devising a reasonable probability distribution to model edge weights. In what follows, we will describe an intuitive and simple process for generating random graphs and assigning random weights to generated edges. This was inspired by \cite{songdechakraiwut2020topological}.

Assume we want to generate a random modular network with $v$ nodes, $K$ modules, an intramodule edge connectivity of $\pi$ and where random edge weights are drawn from distributions $\mu + \sigma Z$ with probability $\pi$ or less for intramodule edges and $\sigma Z$ with probability not exceeding $1$-$\pi$ for intermodule edges. Here, $\mu$ and $\sigma$ are chosen mean and random perturbation (standard deviation) term, and $Z$ is the zero mean Gaussian random variable with unit variance. Also denote the random connectivity probability between any two node, $i$ and $j$ by $\pi_{ij}$. The  algorithm listing (\ref{alg:cap}) is given in the next page.

\subsection{The Graph Laplacian} Consider the graph $\mathcal{X}$, and denote by $i\sim j$ if node $i$ and $j$ are connected, and $w_{ij}$ the edge weight between the two. The graph Laplacian $\mathcal{L}_0 = (l_{ij})$ is defined as 

\begin{equation}
    l_{ij}=
    \begin{cases}
        -w_{ij}, \hspace{.9cm} i \sim j \\
        \sum_{i\ne j} w_{ij}, \hspace{.3cm} i=j\\
        0, \hspace{1.5cm} otherwise
    \end{cases}
    \label{eqn:laplacian-1}
\end{equation}
\noindent
We also denote the degree matrix of the graph by $D$ where $D$ = $Diagonal(d_{i})$, $d_{i} = \sum_j w_{ij}$ and the weight matrix by $W$. Then the Laplacian can be stated as 

\begin{equation}
    \mathcal{L}_0 = D - W
    \label{eqn: laplacian-2}
\end{equation}

\begin{algorithm}[H]
\caption{Generating random modular graphs}\label{alg:cap}
\begin{algorithmic}[1]
\Require $p, K, \pi, \mu, \sigma$
\Ensure $p \ge K, \pi\in [0, 1], \sigma > 0$
\State Divide the nodes evenly among the K modules
\State Pick any two nodes $v_i$ and $v_j$
\If {$v_i, v_j$ belongs to the same module}
    \If {$\pi_{ij}\le \pi$}
        \State $w_{i, j} \sim \mu + \sigma Z$
    \Else
        \State $w_{i, j} \sim \sigma Z$
    \EndIf
\Else
    \If {$\pi_{ij}\le 1 - \pi$}
        \State $w_{i, j} \sim \mu + \sigma Z$
    \Else
        \State $w_{i, j} \sim \sigma Z$
    \EndIf
\EndIf    
\While{$\{v_i, v_j \}$ is unique}
    \State (2) - (14)
\EndWhile
\end{algorithmic}
\end{algorithm}

\noindent
We will now briefly introduce the concept of orientation of a graph that will helps us prove certain properties of $\mathcal{L}_0$.
Consider the graph $\mathcal{X} = (V, E)$, undirected, an orientation of $\mathcal{X}$ is a mapping $f: E \xrightarrow{} V\times V$ assigning a source(s) and target(t) to every edge in $E$, i.e., for every edge $\{ u, v \} \in E$, either $f({u, v}) = (u, v)$ or $f({u, v}) = (v, u)$. The oriented graph $\mathcal{X}^o$ derived from $\mathcal{X}$ by applying the orientation $f$ is $\mathcal{X}^0 = (X, E^0)$, $E^0 = f(E)$.
Denote the incidence matrix associated with $\mathcal{X}^o$ by $B^o$ which is a $|V|\times|E|$ with entries

\begin{equation}
    b_{ij}=
    \begin{cases}
        \sqrt{w_{ij}}, \hspace{.9cm} s(e_j) = v_i \\
        -\sqrt{w_{ij}}, \hspace{.5cm} t(e_j)= v_i \\
        0, \hspace{1.5cm} otherwise
    \end{cases}
    \label{eqn:laplacian-3}
\end{equation}

\noindent
where $s(e_j)$ and $t(e_j)$ are the source and target edges respectively.

\begin{prop}
    Given the incidence matrix $B^o$, weight matrix $W$ and degree matrix $D$, then the following equation holds: $B^o (B^o)^T = D - W$.
    \label{prop:laplaian-d-w}
\end{prop}

\begin{proof}
We observe that $\left[ B^o (B^o)^T \right]_{ij} = b_{[, j]}(b_{[, j]})^T = \sum_{j}b_{i, j}^2$ when $i = j$ and it's simply $-w_{ij}$ when $i\ne j$, hence the conclusion follows.
\end{proof}

\begin{corollary}
The Graph Laplacian $\mathcal{L}_0$ is positive semidefinite.
\label{coll:laplaian-psd}
\end{corollary}

\begin{proof}
For any $x \in \mathbb{R}^v$, consider an orientation of $\mathcal{X}$, then
$$x^T \mathcal{L}_0 x = x^T \left[ B^o (B^o)^T \right] x = \lvert \lvert (B^o)^Tx \rvert\rvert_2^2 \ge 0$$
\end{proof}

\noindent 
We observe that each row of $\mathcal{L}_0$ sums to 0, hence the column vector $\mathbf{1} = \left[1,  1, \dots, 1 \right]_p^T$ is in the null space of $\mathcal{L}_0$, and 0 and $\mathbf{1}$ are eigenvalue and eigenfunctions of $\mathcal{L}_0$ respectively. Let $\lambda_j$ be the eigenvalue associated with the eigenfunctions $\psi_j$ of $\mathcal{L}_0$. Then we have the following system of equations:

\begin{equation}
    \mathcal{L}_0\psi_j = \lambda_j\psi_j
\end{equation}

\noindent
From Corollary $\ref{coll:laplaian-psd}$, we have that $\mathcal{L}_0 \succeq 0 \implies \lambda_j \ge 0$. If we order the eigenvalues, in non-decreasing order, then we have that $\lambda_0 = 0$ and $\psi_0 = \textbf{1}$.
The eigenequation is solvable in the Euclidean space and the solution is given as sine and cosine basis. 
In general, the equation is numerically solved by converting it to the generalized eigenvalue problem. 
Since $\psi_j$ is an eigenfunction, $c\psi_j$ is also an eigenfunctions. To avoid this ambiguity,  we usually normalize the eigenfunctions such that $\psi_i ^{\top} \psi_j = \delta_{ij}$ with Kronecker delta $\delta_{ij}$. For $\lambda_0=0$,  eigenfunctions is trivially $\psi_0=const$. The constant needs to be  $\psi_0=\frac{1}{\sqrt{p}}\mathbf{1}$ to be normalized properly.

\subsection{Hodge Laplacian} Simplicial k-chains ($\mathcal{K}_k$), as discussed earlier, are just a formal sum of oriented collection of simplices. One other thing to note is that, in constructing $\mathcal{K}_k$, once we've chosen a basis for it, a change in orientation is equivalent to flipping the sign of the coefficients. We can equip $\mathcal{K}_k$ with the standard $\mathcal{L}^2$ inner product, which will result in the Hilbert space denoted $\mathbb{H}(\mathcal{K}_k)$. We also introduced boundary maps $\partial_k: \mathcal{K}_k \xrightarrow{} \mathcal{K}_{k-1}$, which is a map between two finite dimensional vector spaces. If we define a basis for each of these operators, we can abstract the chain operations in terms of matrix multiplication which will allow us to have a matrix representation of the boundary operators. Let $\mathbb{B}$  be the matrix representation of the boundary operator $\partial_k$, the Hodge Laplacian is now defined as follows. Denote by $\mathcal{L}_k$ the k$th$ Hodge Laplacian. Then we have that 

\begin{equation}
    \mathcal{L}_k = \partial_{k+1}\partial_{k+1}^* + \partial_{k}^*\partial_{k}
    \label{eqn:hodge1}
\end{equation}
which is equivalent to 
\begin{equation}
    \mathcal{L}_k = \mathbb{B}_{k+1}\mathbb{B}_{k+1}^T + \mathbb{B}_{k}^T\mathbb{B}_{k}
    \label{eqn:hodge2}
\end{equation}
\noindent
in the matrix form, which is mainly used for actual computation. If there is no ambiguity, we will intermix the boundary operator and its matrix form. Here $\partial_{k}^*$ is the adjoint of $\partial_k$. It's straightforward to check that, $\mathcal{L}_k$ is a $|\mathcal{K}_k|\times |\mathcal{K}_k|$ matrix, where $|\mathcal{K}_k|$ is the cardinality of k-simplices in $\mathcal{K}_k$. The Hodge Laplacian has been very useful not only in establishing various theoretical results but also in various applied works \cite{schaub2020random,muhammad2006control}.

We now demonstrate the composition of the matrix $\mathcal{L}_k$. Take as an example, $k = 0$, then $\mathcal{L}_0 = \partial_1\partial_1^* = \mathbb{B}_1\mathbb{B}_1^T$. 
\begin{figure*}[ht]
 \centering
 \includegraphics[width =\textwidth]{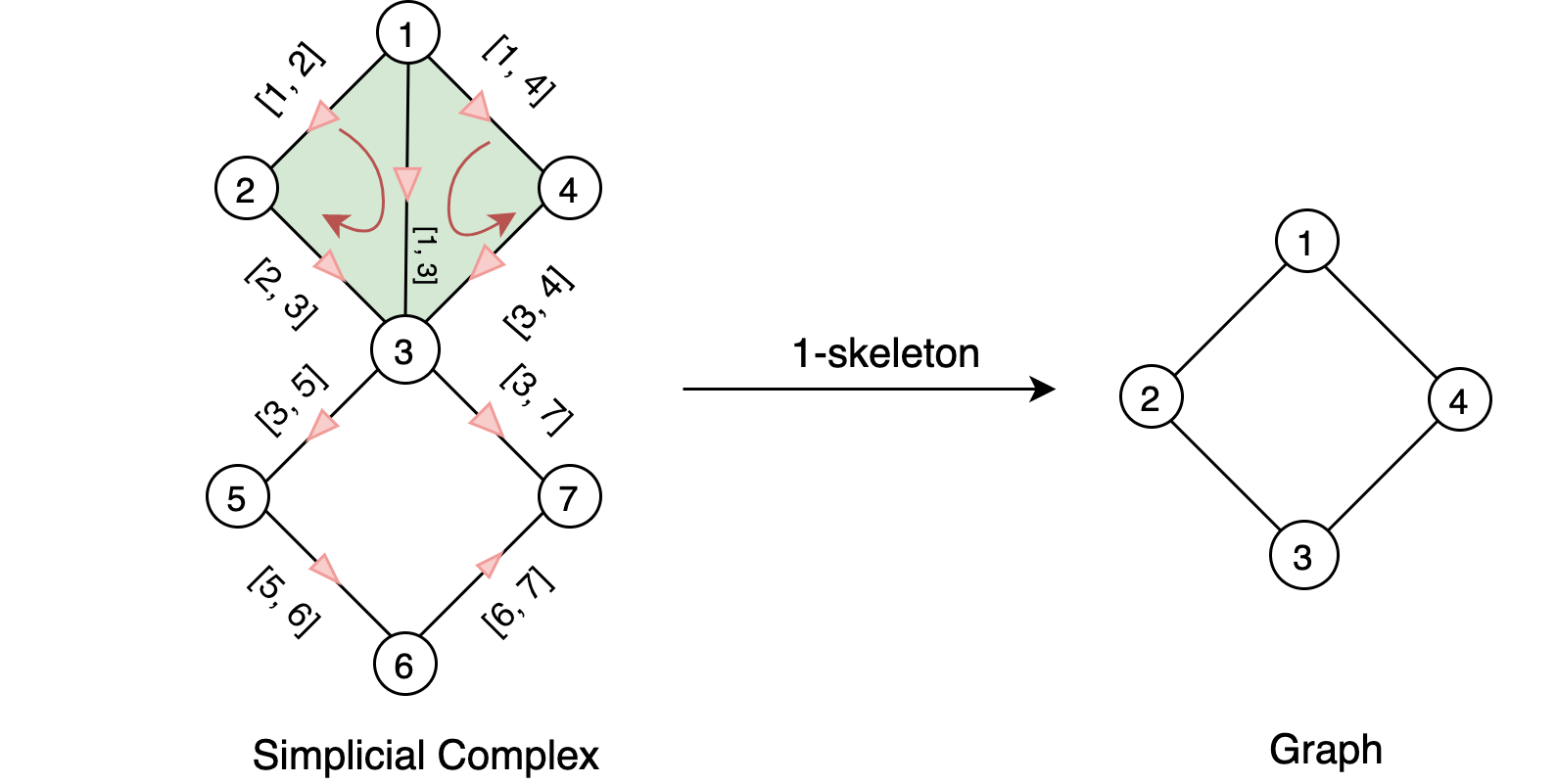}
 \caption{
 A schematic of a 1-skeleton representation of a  simplicial complex $\mathcal{K}$. An oriented simplicial complex $\mathcal{K}$ (left) with each filled in triangle denoting a 2-simplex. A graph corresponding to the 1-skeleton representation of $\mathcal{K}$.}
 \label{fig:1-skeleton}
\end{figure*}
Note that $\partial_1$ maps from 1-simplices (edges) to 0-simplices(nodes), hence $\mathcal{L}_0$ is a variation of the construction associated with equation \eqref{eqn:laplacian-1}, that is
\begin{equation}
    \mathcal{L}_0(ij) =
    \begin{cases}
        -1, \hspace{.9cm} v_i\sim v_j  \\
        deg(v_i), \hspace{.3cm} i = j\\
        0, \hspace{1.3cm} otherwise
    \end{cases}
    \label{eqn:hodge-laplacian-0}
\end{equation}
Figure~\ref{fig:1-skeleton} provides an illustration of such a mapping.
We can check that the matrix representation $\mathbb{B}_1\mathbb{B}_1^T$ will have entries corresponding to the unweighted version in Proposition \ref{prop:laplaian-d-w}. An alternative construction of the Laplacain $\mathcal{L}_k$ for $k>0$ is as follows. Two simplices $\sigma_i$, $\sigma_j$ are said to be upper adjacent if both are faces of some $k+1$-simplex in $\mathcal{K}$. Denote this adjacency by $\sigma_i \wedge \sigma_j$. The upper degree of a k-simplex $(deg_u(\sigma))$ is the number of $k+1$-simplices in $\mathcal{K}$ of which $\sigma$ is a face. The above generalizes the concept of degree and adjacency of graphs, and note that analogous definitions holds for the lower degree and upper adjacency. It can be shown that \cite{muhammad2006control}
\begin{equation}
    \mathcal{L}_k = D - A_u + (k+1)I_{N_k} + A_l
    \label{eqn:laplace_components}
\end{equation}
where $D$ is the diagonal degree matrix, $A_u(A_l)$ is the upper(lower) degree matrix, $I_{N_k}$ is the identity matrix.

It's now straightforward to compute the various Laplacian matrices for the example shown in Figure~\ref{fig:1-skeleton}.
$\mathcal{L}_0 = \mathbb{B}_1\mathbb{B}_1^T$ will be a $7\times 7$ matrix and $\mathcal{L}_1 = \mathbb{B}_2\mathbb{B}_2^T + \mathbb{B}_1^T\mathbb{B}_1^T$ will be a $9 \times 9$ matrix and $\mathcal{L}_2 = \mathbb{B}_2^T\mathbb{B}_2 $ will be a $2\times 2$ matrix.

\subsubsection{Cycles} We now extend the definition of a cycle under 1-skeleton to any  simplicial complex of arbitrary order. In previous section, we defined the $k$-th homology group as $\calH_k = \calZ_k / \calB_k$ where $\calZ_k$ is the $k$-cycle group and $\calB_k$ is the $(k+1)$-boundary group with $\calB_k \subset \calZ_k$. The rank of the homology groups gives the Betti numbers $\beta_k = rank(\calH_k)$ \cite{hatcher2002algebraic, topaz2015topological}. Since the $k$-th homology group contains k-dimensional cycles which differ by only homological cycles, the $\beta_k$ counts the number of k-cycles. The basis associated with $\calH_0$, $\calH_1$ and $\calH_2$ corresponds to connected components, holes and cavities respectively.
 We will now motivate the concept of persistence of k-cycles through persistent homology. Let $\calK$ be a simplicial complex, and analogous to graph filtration, consider a filtration $\{\calK_{\epsilon_k}\}_{\epsilon_k \in\mathbb{R}}$ where for $\epsilon_{k_1} \le \epsilon_{k_2}$, we have that $\calK_{\epsilon_{k_1}} \subset \calK_{\epsilon_{k_2}} \subset \calK$.  The filtration of $\calK$ induces a homomorphism between the homology groups, i.e. $\calH_k\left(\calK_{\epsilon_{k_1}}\right) \xrightarrow[]{} \calH_k\left(\calK_{\epsilon_{k_2}}\right)$, which allows us to track the birth of homology classes over the span of the filtration to their death point. Persistence modules, which uniquely decompose into persistence intervals \cite{zomorodian2005computing} is a tool for tracking this evolution. Let $\calP_k$ denote the persistence module, and for a $k$-cycle $\calC \in \calP_k$, denote the point it first appeared (birth) as $b(\calC)$ and the point it disappeared (death) as $d(\calC)$. Then the lifespan of this cycle can be represented as $\left[b(\calC),  d(\calC)\right)$, the persistence interval. A finite collection of the persistence interval $\left[b(\calC),  d(\calC)\right)$ can be summarized in the form of a \textit{barcode}, which can intuitively be thought of as the persistence analogue of Betti numbers.

\subsubsection{Diffusion on Simplicial Complexes} Consider a 1-skeleton $\mathcal{X}  = (V, E)$, where $V$ is the set of $p$ nodes and $E$ is the edge set. The Laplacian $\mathcal{L}_0$, similar to the previous discussion can be written as the difference between the degree matrix ($D$) and the adjacency matrix ($A$). The concept of diffusion on k-skeletons is to generalize the notion of frequency and oscillation to any k-skeleton (could be irregular) via a linear operator.
Define $g$ to be a 1-skeleton signal, that's a function defined on the vertices $V$. We can view g as a measure of the amount of substance that moves from node $i$ to $j$ or vice-versa. The change (gradient) in this function from node $i$ to $j$  will have the form
\begin{equation}
    \triangledown g(i) = c(g(j) - g(i))
    \label{eqn:diff_function}
\end{equation}
where the constant $c$ can be viewed as the diffusion rate across the edge. Then $\triangledown g$ is a mapping from $V$ to $E$. From the previous discussion on the eigenvalues and eigenfunctions of the Laplacian $\mathcal{L}_0$, the collection of $\psi_j$ forms an orthogonal basis for the 1-skeleton signal $g\in \mathbb{R}^p$, often defined on the Hilbert space.

Let $\mathcal{F}(\lambda_j)$ denote the Fourier transform of $g$, which is an expansion of $g$ interms of the eigenvectors of the Laplacian $\calL_0$. Then
\begin{equation}
    \mathcal{F}(\lambda_j) = \langle g, \psi_j \rangle = \sum_{i \in V}g(i)\psi_j(i)
    \label{eqn:ft}
\end{equation}
Since the eigenfunctions form an orthonormal basis, an inverse Fourier transform exists and is given by:
\begin{equation}
    g(i) = \sum_{j = 0}^{p-1} \mathcal{F}(\lambda_j) \psi_j(i)
    \label{eqn:ft_inverse}
\end{equation}
Assume we have initial observed data $f$ defined on node set $V$. Take the diffusion equation 
 on a 1-skeleton signal $g$ to be
 \begin{equation}
    \frac{\partial g(t)}{\partial t} = \sminus\mathcal{L}_0 g(t)
    \label{eqn:prop-gs-1}
 \end{equation}
with initial condition $g(t=0)=f$. The unique solution to this equation is given by 
 \begin{equation}
    g(t) = \sum_{j=0}^{p-1} e^{-\lambda_j t} f_i \psi_j
    \label{eqn:prop-gs-2}
 \end{equation}
where $f_j = f^{\top} \psi_j$ is the Fourier coefficient with respect to $\psi_j$ \cite{huang2019fast,RICAUD2019474}.

We can generalize this result to k-skeletons. Assume $g(t)$ to be functional data with some parameter $t$ at a generic simplex $\sigma$ (a collection of k-simplices). It is expected such functional data is highly noisy and requires denoising, which will be done through the heat equation. Before we formulate the generalized model, we state some important properties associated with the Hodge Laplacian $\mathcal{L}_k$ for a general $k$. There are quite a number of important properties of the eigenvectors and eigenvalues of $\mathcal{L}_K$ \cite{barbarossa2018learning}. One important property is the symmetricity of $\mathcal{L}_k$, which also implies its diagonalizable. Ordinarily, been symmetric is the property of operators and not just the associated matrix, but as indicated earlier we will intermix the boundary operator and the matrix form. Henceforth, properties defined will be assumed to be on the matrix form of $\mathcal{L}_k$.

\begin{prop}
Given the kth Hodge Laplacian $\mathcal{L}_k$, denote its set of eigenvectors by $\psi = \left\{\psi_1, \tdots, \psi_{|\mathcal{K}_k|}\right\}$ with associated eigenvalues $\lambda = \left\{\lambda_1, \tdots, \lambda_{|\mathcal{K}_k|}\right\}$. Then $\psi$ forms an orthonormal basis of $\mathbb{R}^{|\mathcal{K}_k|}$.
\label{prop:orthonormal-basis-of-Lk}
\end{prop}

\begin{proof}
The conclusion follows from the symmetricity of $\mathcal{L}_k$ if we can show that the eigenvectors are mutually orthogonal.

Choose any two eigenvectors $\psi_i, \psi_j$, taking the inner product we have
$$\langle \mathcal{L}_k\psi_i, \psi_j\rangle = \langle \psi_i, \mathcal{L}_k^T\psi_j\rangle$$
Since $\mathcal{L}_k$ is symmetric, let $\lambda_i \ne \lambda_j$ be the corresponding eigenvalues of $\psi_i, \psi_j$. Usually we choose $\psi_i, \psi_j$ such the uniqueness condition is satisfied. Then we have
\begin{align}
    \lambda_i \langle \psi_i, \psi_j \rangle &
    \begin{aligned}[t]
        &=  \langle \lambda_i \psi_i, \psi_j \rangle\\
        &= \langle \mathcal{L}_k \psi_i, \psi_j \rangle\\
        &= \langle\psi_i,  \mathcal{L}_k^T \psi_j \rangle\\
        &= \langle\psi_i,  \mathcal{L}_k \psi_j \rangle\\
        &= \lambda_j \langle\psi_i,  \psi_j \rangle\\
    \end{aligned}
\end{align}
Since 
$$\lambda_i \ne \lambda_j \implies (\lambda_i - \lambda_j) \langle \psi_i, \psi_j \rangle = 0 \implies \psi_i \perp \psi_j$$

we have that the set of eigenspaces are mutually orthogonal which implies when taken together, will form an orthonormal  subset of $\mathbb{R}^{|\mathcal{K}_k|}$. Symmetric matrices are diagonalizable and since $\mathcal{L}_k$ is symmetric, this orthonormal set will be a basis.
\end{proof}

In essence, what this proposition demonstrates is that, two distinct eigenvalues will yield orthogonal eigenvectors, and it can be verified that this still holds even when we have a single unique eigenvalue.

Define the {\em simplex Fourier transform} $\mathcal{F}$ as a mapping from $\mathcal{K}_k$ to $\mathbb{R}$. 
Let $\mathcal{F}(\lambda_j)$ denote the simplex Fourier transform of functional data $g$ defined over $\mathcal{K}_k$ at $\lambda_j$:
\begin{equation}
    \mathcal{F}(\lambda_j) = \sum_{\sigma \in \mathcal{K}_k}g(\sigma)\psi_j(\sigma).
    %\label{eqn:ft}
\end{equation}

\begin{prop}
Let $\mathcal{K}_k$ be a collection of such $k$-simplices with cardinality $|\mathcal{K}_k|$. Let $f$ be the observed functional data data defined over $\mathcal{K}_k$. The diffusion equation over the simplices  
\begin{equation}
    \frac{\partial g(t)}{\partial t}  = \sminus\mathcal{L}_k g(t)
    \label{eqn:higher-diff}
\end{equation}
\noindent
with initial condition $g(t=0) = f$ has the unique solution given by
\begin{equation}
     g(t) = \sum_{j=0}^{|\mathcal{K}_k| -1}e^{-\lambda_jt}f_j \psi_j
    \label{eqn:simplex-diffusion-equation},
\end{equation}
\label{prop:diffusion-on-simplices}

where $\lambda_j$ and $\psi_j$ are the $j\text- th$ eigenvalue and eigenvector of the $k\text- th$ Hodge Laplacian  $\mathcal{L}_k$, and $f_j= f^T\psi_j$ is the expansion coefficient.

\end{prop}

\begin{proof}
We showed in Proposition \eqref{prop:diffusion-on-simplices} that the eigenfunctions of $\mathcal{L}_k$ forms an orthonormal basis. Thus it applies that the inverse Fourier transform exists and is given similarly by \eqref{eqn:ft_inverse} with 0-simplices(nodes) replaced with k-simplices. It then follows that
$$\mathcal{F}( \frac{\partial g(t) }{\partial t }) = \mathcal{F}\left(\sminus \mathcal{L}_kg(t) \right),\quad \frac{\partial \mathcal{F}\left( g(t) \right)}{\partial t}   = \sminus \mathcal{L}_k\mathcal{F}\left( g(t) \right)$$
from which the general solution is given by $$\mathcal{F}(g(t)) = e^{\sminus t \lambda_j}\mathcal{F}_0(g(t))$$ 
where $\mathcal{F}_0(g(t))$ is the Fourier transform at the initial condition $g(t = 0) = f$ and $\mathcal{F}_0(g(t)) = f_j = f^T\psi_j$ been the Fourier coefficients. Using (\ref{eqn:ft_inverse}) we have that 
\begin{equation}
     g(t) = \sum_{j=0}^{|\mathcal{K}_k| -1}e^{-\lambda_jt}f_j\psi_j
    \label{eqn:simplex-diffusion-solution}
\end{equation}
\end{proof}

\subsection{Heat kernel over simplices}

Given the simplicial complex $\calK_k$, let $f$ be the functional data defined over the $k$-simplices in $\calK_k$. From Proposition \ref{prop:orthonormal-basis-of-Lk}, the eigenfunctions of the $k$-$th$ Hodge Laplacian  $\calL_k$ are given by $\psi_0, \dots , \psi_{|\calK_k|-1}$ with the corresponding eigenvalues $\lambda_0 \le \dots \le \lambda_{|\calK_k|-1}$. The diffusion $ g(t, \sigma)$ of initial functional data $f(\sigma)$ on $\calK_k$ after time $t$ is given by 
\begin{equation}
    \frac{\partial}{\partial t}g(t, \sigma) = -\calL_k g(t, \sigma), \hspace{0.3cm} g(t = 0) = f (\sigma)
    \label{eqn:heat-equation-2}.
\end{equation}

For $\sigma_1, \sigma_2 \in \calK_k$, the heat kernel is defined as
\bq K_t(\sigma_1, \sigma_2) =  \sum_{j = 0}^{|\calK_k|-1} e^{-\lambda_j}\psi_j(\sigma_1)\psi_j(\sigma_2) \label{eqn:
heat-kernel}, \eq
where $t$ is the kernel bandwidth. With the exception of a few geometries, the heat kernel $K_t$ does not admit an explicit analytic form since the eigenfunctions $\psi_j$ are not often  given analytically \cite{chung2015unified}. A proposed workaround is to study the small bandwidth behavior via the parametrix expansion, the asymptotic expansion of heat kernel \cite{rosenberg1997laplacian}. When the geodesic distance $d(\sigma_1, \sigma_2)$ between simplices $\sigma_1$ and $\sigma_2$ is small enough,  we have the following parametrix expansion  \cite{chung2015unified}
\bq K_t(\sigma_1, \sigma_2) = \left( 4\pi t \right)^{\frac{\sminus (k+1)}{2}}\exp{\sminus \frac{d^2(\sigma_1, \sigma_2)}{4t^2}}[1 + \mathcal{O}(t^2)] \label{eqn:parametrix-expansion}, \eq
which approximately follows the Gaussian kernel defined on the tangent space near $\sigma_1$ and $\sigma_2$. The tangent space is Euclidean.

Note that each simplex $\sigma_j \in \calK_k$ is an ordered tuple of $k$+$1$ elements of the metric space. Denote these elements by $\sigma_j = \{ \sigma_{j0}, \dots, \sigma_{jk} \}$. If the following holds
\bse \frac{1}{2}[d^2(\sigma_{j0}, \sigma_{jj^\prime}) + d^2(\sigma_{j0}, \sigma_{jj^{\prime\prime} }) - d^2(\sigma_{jj^\prime}, \sigma_{jj^{\prime\prime} })] > 0 \ese
then the metric space is flat, in which case the dimension of the metric does not exceed a specified Euclidean $k$+$1$-space \cite{morgan1974embedding}. The metric space can then be embedded in the Euclidean space, and the heat kernel becomes a Gaussian kernel \cite{chung2015unified,morgan1974embedding}.
If we convolve $K_t$ with the initial functional data, we can algebraically show that the convolution is solution to \eqref{eqn:heat-equation-2} \cite{chung2015unified,huang2019fast}:
\bq K_t*f (\sigma) = \sum_{j=0}^{|\calK_k|-1}e^{-\lambda_j t}f_j\psi_j (\sigma) \label{eqn:gaussian-kernel-convolve} \eq

\subsection{Simulation Study} This section will demonstrate the diffusion process for 0-simplices and 1-simplices. We first generate a graph with random perturbation. Algorithm \ref{alg:cap} can be adapted to generate this graph. The Hodge's $\mathcal{L}_0$ and $\mathcal{L}_1$ can now be computed.
In both cases, we will follow the Algorithm listed below which implements Proposition \ref{prop:diffusion-on-simplices}.

\begin{breakablealgorithm}
%\begin{algorithm}[H]
\caption{Diffusion on simplicial complex}
\label{alg:diffons-simplices}
\begin{algorithmic}[1]
\Require $k, i,  f, T, \mathcal{K}_k$
\Ensure $k \in [0, 2]$, $i \in (0, T)$, $T>0$
\If{k=0}
    \State Compute the adjacency matrix $A$ and degree matrix $D$
    \State Compute $\mathcal{L}_0 = D - A$ and $\bl$, $\bpsi$ using $ \mathcal{L}_0 \bpsi = \bpsi\bl$
    \State \Return $\bl$, $\bpsi$ 
\ElsIf{k=1}
    \State Obtain set of all 1(0)-simplices $E(V)$, and compute $|E|$ and $|V|$
    \State Initialize matrices ${B_1}_{|V| \times |E|}$ and ${B_2}_{|E| \times \sim}$ \Comment{$\sim \implies$ dynamic column size}
    \For {$e \in E$}
        \State Construct the two nodes of $e$: $e_1$, $e_2$ and a row index $rn_1 = `e_1\sminus e_2`$
        \State Set $B_1[e_1, rn] \xleftarrow{} {\sminus 1}$, $B_1[e_2, rn] \xleftarrow{}{1}$ 
        \For{$p \in V$}
            \If{$\{e_1, p\}$ is 1-simplex \& $\{e_2, p\}$ is 1-simplex}
                \State Get indexes $rn_2 = `e_1\sminus p`, rn_3 = `e_2\sminus p`$, $cn = `e_1\sminus e_2-p`$
                \State Choose an orientation  
                \State Set $B_2[rn_1, cn] \xleftarrow{} {1}, B_2[rn_2, cn] \xleftarrow{}{1}, B_2[rn_3, cn] \xleftarrow{} {\sminus 1}$
            \EndIf
        \EndFor
    \EndFor
    \State Compute $\mathcal{L}_1 = B_{2}B_{2}^T + B_{1}^TB_{1}$ and $\bl, \bpsi$ using $ \mathcal{L}_1 \bpsi = \bpsi\bl$
    \State \Return $\bl$, $\bpsi$
\Else
    \State
    \State Obtain $k\text-1(k)$-simplices and $k(k+1)$-cycles
    \State Initialize matrices $B_k \xleftarrow[]{} |\calK_{k-1}| \times |\calK_k|$, and $B_{k+1} \xleftarrow[]{} |\calK_{k}| \times |\calK_{k+1}|$ 
    \For {$\sigma_{k-1} \in \calK_{k-1}$}
        \For {$\sigma_{k} \in \calK_{k}$}
            \If{$\sigma_{k-1}$ is a face of $\sigma_k$}
                \State Construct $rn_1 \xleftarrow[]{}$ row index of $\sigma_{k-1}$ in $B_{k}$
                 \State Construct $cn_1 \xleftarrow[]{}$ column index of $\sigma_{k}$ in $B_{k}$
                 \State Set $B_{k}[rn_1, cn_1] \xleftarrow[]{} 1$
            \EndIf
            \For {$\sigma_{k+1} \in \calK_{k+1}$}
                \If{$\sigma_{k}$ is a face of $\sigma_{k+1}$}
                    \State Construct $rn_2 \xleftarrow[]{}$ row index of $\sigma_{k}$ in $B_{k+1}$
                     \State Construct $cn_2 \xleftarrow[]{}$ column index of $\sigma_{k+1}$ in $B_{k+1}$
                     \State Set $B_{k+1}[rn_2, cn_2] \xleftarrow[]{} 1$
                \EndIf
            \EndFor
        \EndFor
    \EndFor
    \State Compute $\mathcal{L}_k = B_{k+1}B_{k+1}^T + B_{k}^TB_{k}$ and $\bl, \bpsi$ using $ \mathcal{L}_k \bpsi = \bpsi\bl$
    \State \Return $\bl$, $\bpsi$
\EndIf
\State Transform $f$ to coordinate system of eigenvectors $f_0 =\bpsi^Tf$
\For {$t = 0$ $to$  $T$, increment = i}
    \State Compute $\mathbf{g} =  f_0 \odot exp(\sminus \bl t)$ \Comment{$\odot$  element-wise multiplication}
    \State Transform back to original coordinates $\mathbf{g} = \bpsi \mathbf{g}$
\EndFor
\State \Return $\mathbf{g}$
\end{algorithmic}
%\end{algorithm}
\end{breakablealgorithm}

In the preceding algorithm, we note that both diffusion on 0-simplices and diffusion on 1-simplices can be conducted based of line (5) - (24). But line (1) - (4) reduces greatly, the computation time when only diffusion on 0-simplices is required.

To illustrate diffusion over 0-simplices, we simulated a 1-skeleton (graph) of 200 0-simplices(nodes) in two modules with intramodule connectivity probability of $\pi = 0.19$, perturbed the connectivity with noise drawn from a Gaussian with mean $1$ and standard deviation $0.25$. The initial node data were drawn from a multinomial distribution with equal probability and scaled by a factor of five (5). A diffusion bandwidth of $0$, $0.5$, and $1.0$ units were used to smoothed over the nodes. Figure~\ref{fig:node-smoothing} shows the results from the 0-simplices smoothing.
\begin{figure*}[ht]
 \centering
 \includegraphics[width =\textwidth]{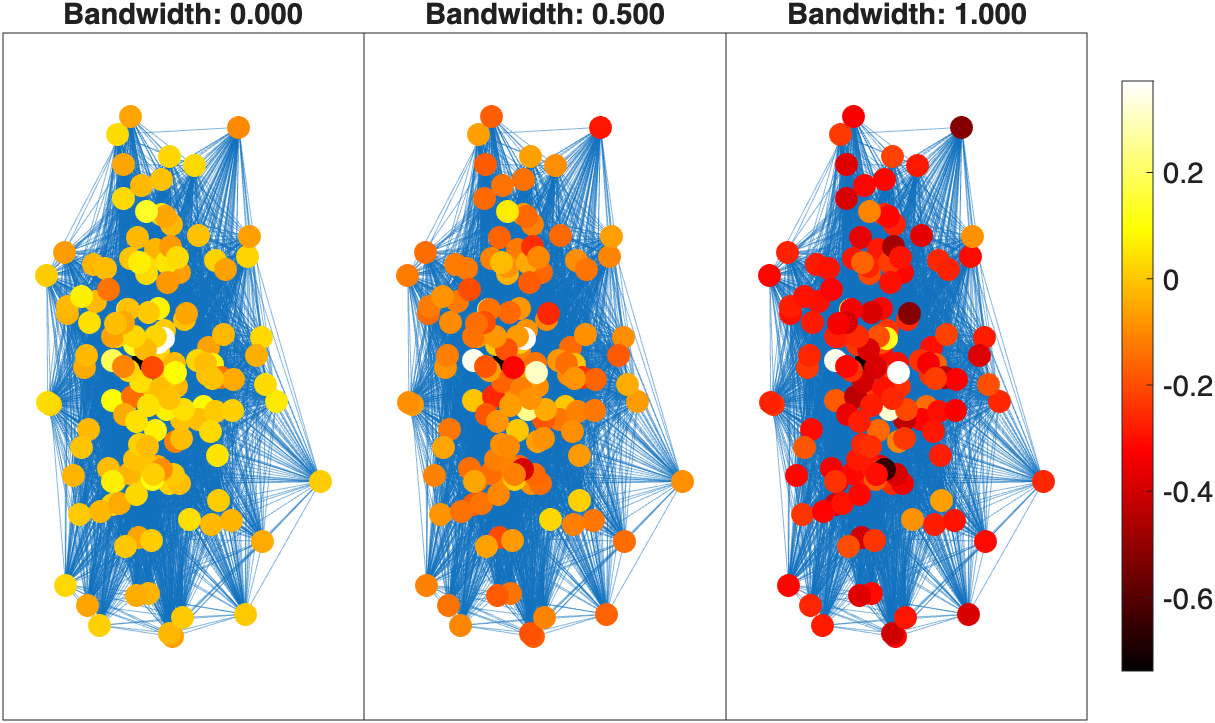}
 \caption{An illustration of the simulated and smoothed data. On the far-left is the collection of 150 0-simplices with the color gradient indicating the initial concentration at each of the nodes. We notice the color spectrum is widely varying as expected. In the middle and far-right is the resulting 0-simplices after smoothing via the heat diffusion equation. Its apparent that there is a greater uniformity in the distribution of the color gradient compared to the unsmoothed one. This signals the denoising of the initial functional data defined on the nodes.}
 \label{fig:node-smoothing}
\end{figure*}

Similarly, to illustrate diffusion on 1-simplices, we simulated a collection of 1-simplices in two blocks and connected with an intramodule connectivity probability of $\pi = 0.95$. All other parameter settings remain the same as in the 0-simplex diffusion setup. The results are shown in Figure~\ref{fig:edge-smoothing}.
\begin{figure*}[ht]
 \centering
 \includegraphics[width =\textwidth]{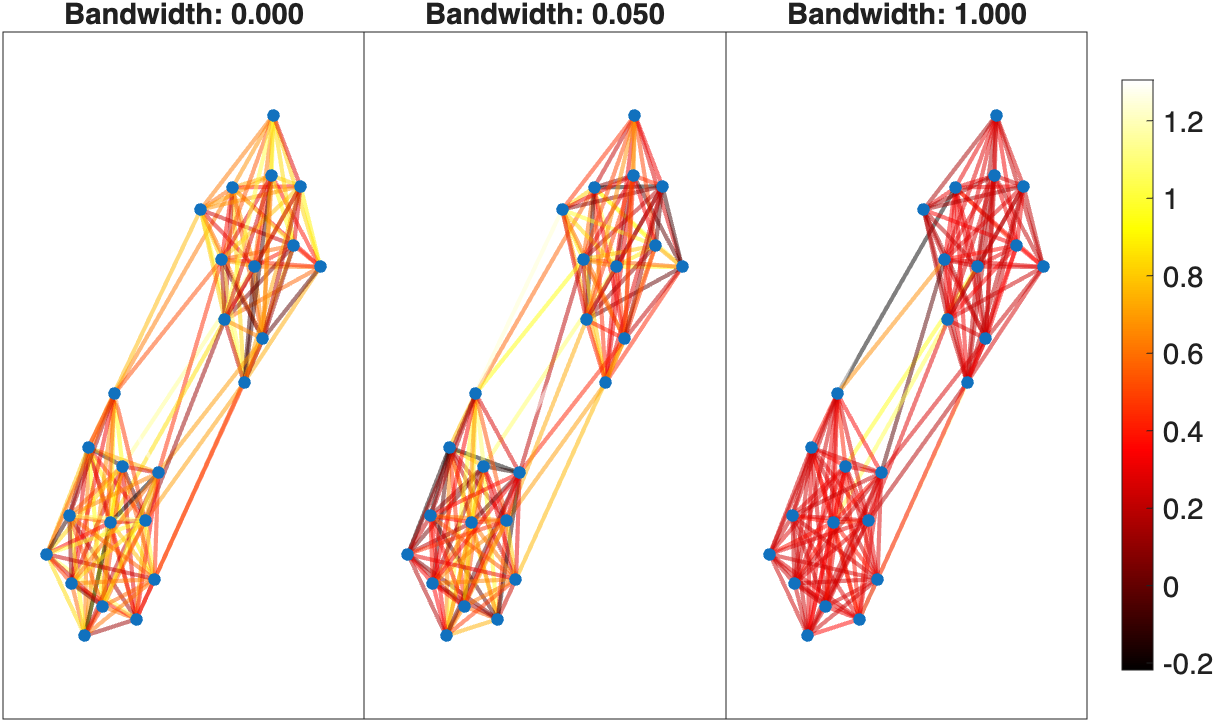}
 \caption{This illustrate the diffusion process on 1-simplices. On the far-left is the collection of 1-simplices with the initial unsmoothed functional data. Again the color gradient is quite varying. The middle and far-right graph depicts the output when the data is denoised. We observe that color gradient across the 1-simplices is relatively uniform compared to the unsmoothed data.}
 \label{fig:edge-smoothing}
\end{figure*}
The significance of these simulations underscores the power of diffusion via the heat equation to denoise functional data defined either on 0-simplices(nodes) or 1-simplices(edges). The next section generalizes this result for general k-simplices. 

\subsection{Representation of k-cycles}

The spectral decomposition of the k-th Hodge Laplacian $\calL_k$ is given by \bse \calL_k \U_k = \U_k \bL_k   \ese
where $\bL$ is a diagonal matrix with the diagonal entries been the eigenvalues of $\calK_k$ with corresponding eigenvectors been the columns of the matrix $\U_k$. The Betti number $\beta_k$ can now be obtained as the multiplicity of the zero eigenvalues associated with the decomposition. In the special case of $k = 1$, it has been shown that, the $m$-th 1-cycle can be written as a linear combination of edges and the vectors of $\U_1$ corresponding to the zero eigenvalues. A straightforward generalization of this result for any arbitrary $k$ can be constructed. 

We can characterize the $k$-cycles as a linear combination of the vectors in $\U_k$ corresponding to the zero eigenvalues. Each eigenvector corresponding to the zero eigenvalue represents $k$-cycles. Then the $m$-th $k$-cycle is given by 
\bq \calC^m = \sum_{\sigma_k \in \calK_k} u_{0}(k) \sigma_k \label{eqn:k-cyles-all}\eq
where $\sigma_k$ is a $k$-simplex with the corresponding $k$-th entry in $\U_k$, and $u_{0}(k)$ is an eigenvector (a column of $\U_k$) corresponding to a zero eigenvalue.

Similar to the realization in \cite{anand2021hodge}, even k-simplices that do not form part of the k-cycle are considered in (\ref{eqn:k-cyles-all}). The way to solve this problem for the special case of $k = 1$ is to sequentially construct the 1-cycles drawing from the theory of spanning trees \cite{anand2021hodge, busaryev2012annotating}. This approach is not feasible for the case where $k > 1$. But we can have an analogous construction that retains the core property that the entries of the eigenvectors will be zero for k-simplices that are not part of the cycle.

Recall in Algorithm \ref{alg:diffons-simplices}, we defined the boundary matrix $\B_k$ to have dimension $|\calK_{k-1}| \times |\calK_{k}|$. The columns of $B_k$ corresponds to the $k$-$1$ boundary of the $k$-simplex $\sigma_j$, $j = 1, \dots, |\calK_{k}|$. If we let $r = rank(\B_k)$, then there exists some permutation matrix $\P$ such that 
\bq \B_k = \Q[ \I_r | \R] \P \eq
where $[ \boldsymbol{.} | \boldsymbol{.} ]$ denotes the horizontal concatenation of two matrices \cite{busaryev2012annotating}. The matrix $\I_r$ is an $r\times r$ identity matrix, $\Q$ is a $|\calK_{k-1}| \times r$ matrix with linearly independent columns, $\R$ is a $r\times |\calK_{k}| - r$ matrix. Note that if $\B_k$ has full column rank, then $\B_k = \Q$ in which case the columns of $\B_k$ form a basis. If we assume otherwise, then the columns of $\Q$ form a basis for the column space of $\B_k$. We also note that $\P\P^T = I$, which implies the permutation matrix can act on the columns of $\B_k$ such that the the first $r$-columns of $\B_k\P$ represents $Q$.

Now we observe that the first $r$ $k$-simplices $T = \{ \sigma_j, \dots, \sigma_r \}$ will have linearly independent boundaries, hence any subset $\{ \sigma_j \}_{j=1}^{k}, k \le r$ will not form a $k$-cycle. Let $\calK_k^\prime = \{ \sigma: \sigma \in \calK_k, \sigma \not\in T \}$. Analogous to the 1-skeleton case, adding any simplex $\sigma \in \calK_k^\prime$ to $T$ will form a k-cycle. The cycle formed by $\sigma_{r +j} \cup T, j = 1 , \dots, |\calK_k|$ only contains one simplex from $\calK_k^\prime$. This now allows to rewrite $(\ref{eqn:k-cyles-all})$ as

\bq  \calC^m = \sum_{\sigma_k \in \calK_k^\prime} \mu_0(k)\sigma_k \label{eqn:k-cycles-relevant} \eq
where we assume that for any $\sigma \in T$, the cycle formed by $T \cup \sigma$ is set to 0.

\begin{figure*}[ht]
     \centering
     \begin{subfigure}[t]{0.325\linewidth}
         \centering
         \includegraphics[width=\textwidth]{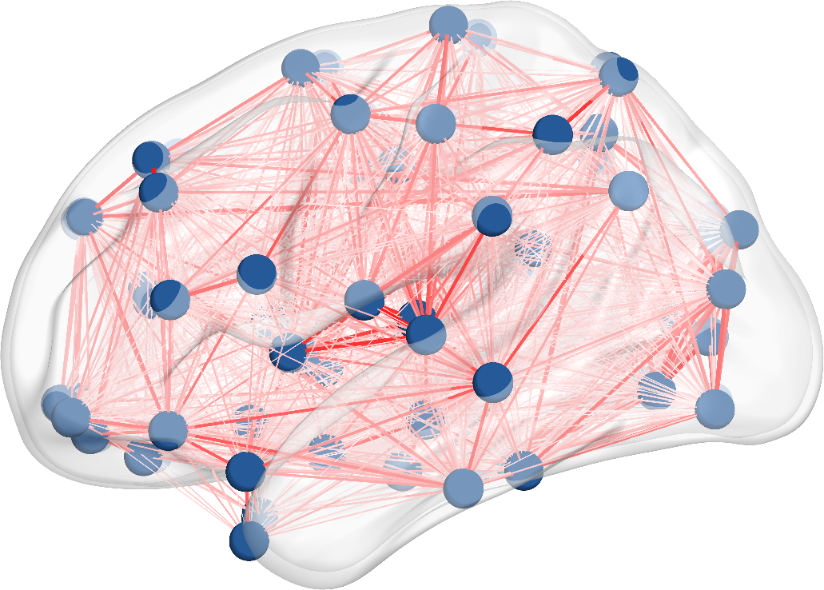}
     \end{subfigure}
     \hfill
     \begin{subfigure}[t]{0.325\linewidth}
         \centering
         \includegraphics[width=0.75\textwidth,height=0.75\textwidth]{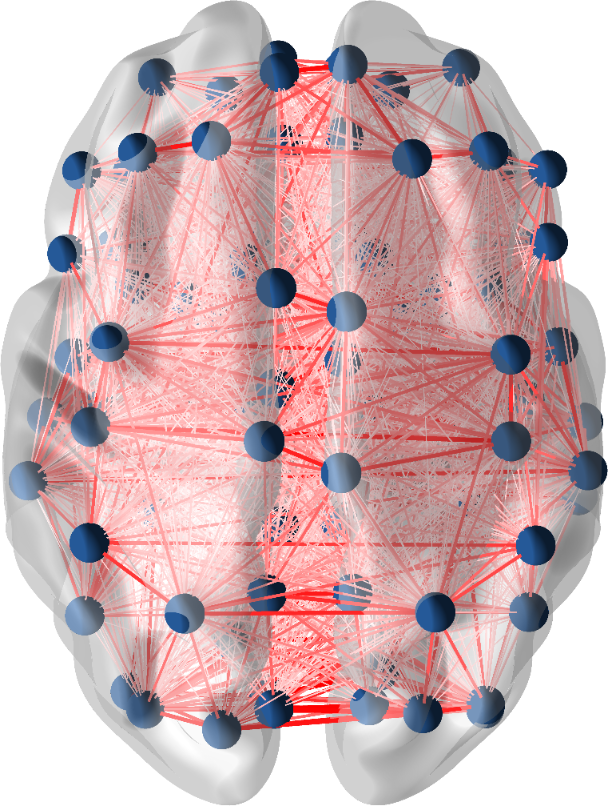}
     \end{subfigure}
     \hfill
     \begin{subfigure}[t]{0.325\linewidth}
         \centering
         \includegraphics[width=\textwidth]{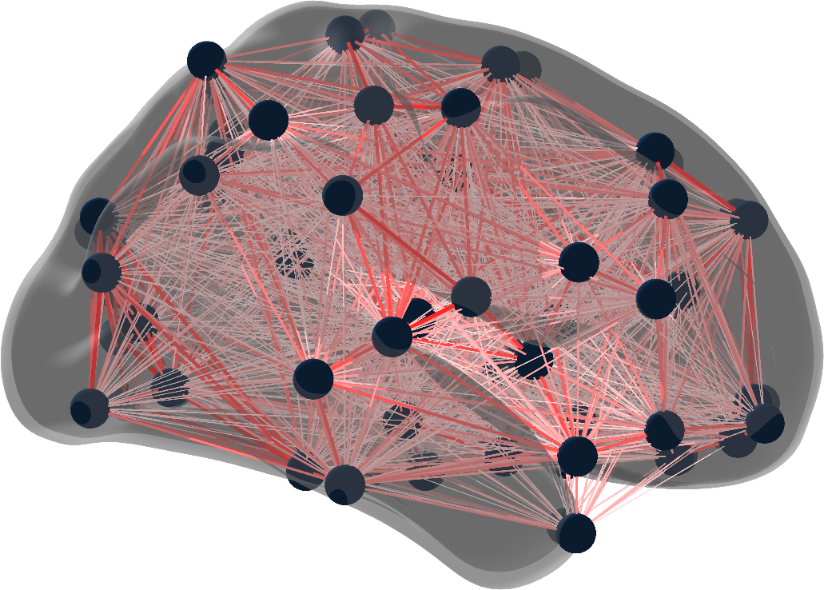}
     \end{subfigure}
     
    \par\vspace{0.5cm}
    
     \centering
     \begin{subfigure}[t]{0.325\linewidth}
         \centering
         \includegraphics[width=\textwidth]{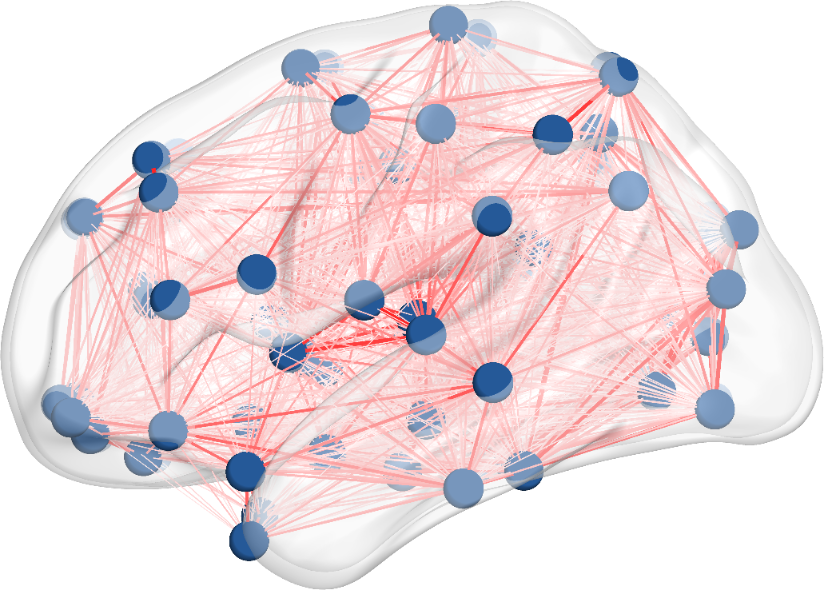}
     \end{subfigure}
     \hfill
     \begin{subfigure}[t]{0.325\linewidth}
         \centering
         \includegraphics[width=0.75\textwidth,height=0.75\textwidth]{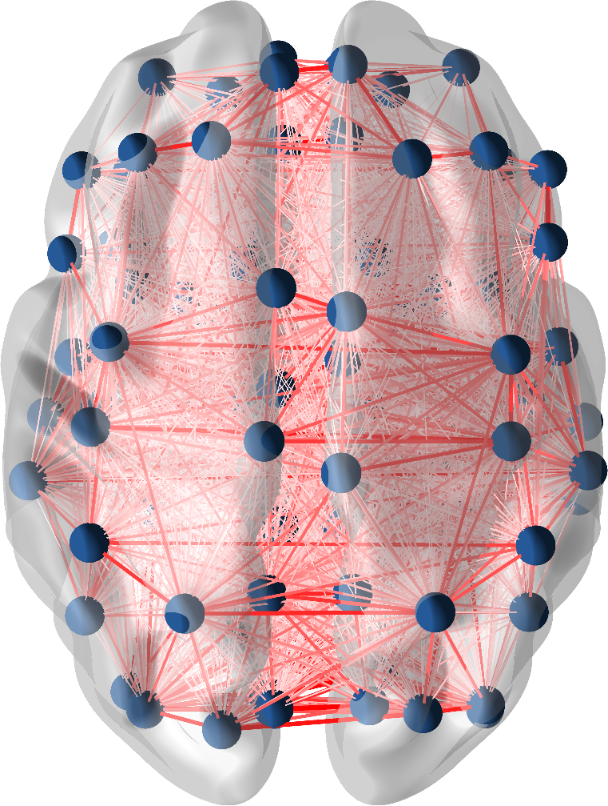}
     \end{subfigure}
     \hfill
     \begin{subfigure}[t]{0.325\linewidth}
         \centering
         \includegraphics[width=\textwidth]{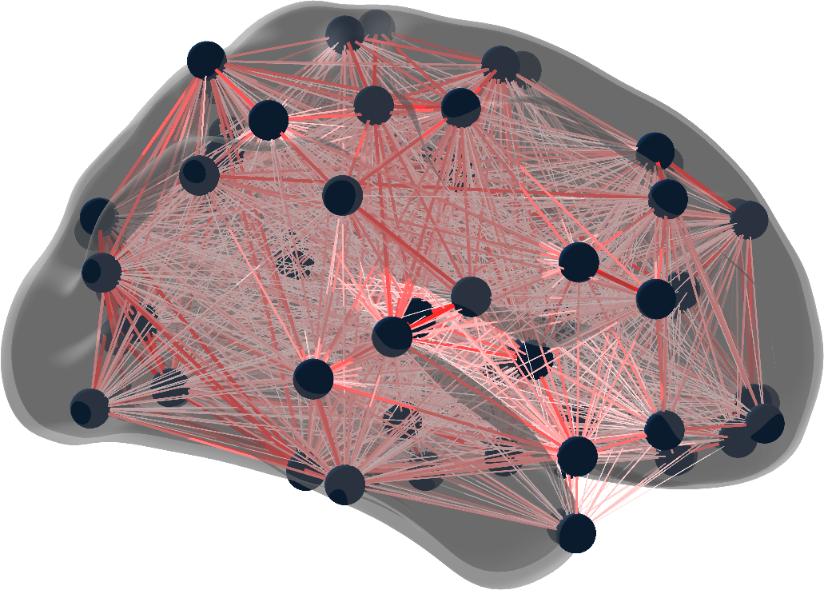}
     \end{subfigure}
      \begin{subfigure}[t]{1\linewidth}
         \centering
         \includegraphics[width=0.75\textwidth]{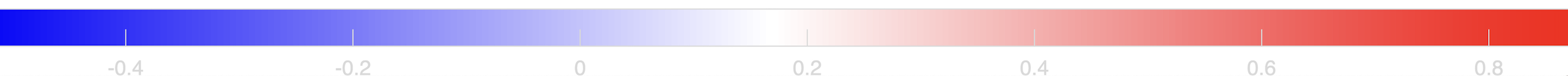}
     \end{subfigure}
    \caption{Top: Group-mean functional connectivity of users on a glass brain. Relatively few negative connections were observed. Bottom: Group-mean functional connectivity of non-users on a glass brain. Relatively few negative connections were observed. The dense representation of the connectivities gives no visibly discernible pattern.}
    \label{fig:mta-raw-non-use}
\end{figure*}

\section{Application to Human Brain Networks} \label{sec:app}
The proposed framework is applied to real data with an edge-domain analysis of human resting-state fMRI. The dataset is provided as parcel-wise functional connectivity networks (see figure~\ref{fig:mta-raw-non-use}), so we focus on edge smoothing, which regularizes the connectivity signal on the graph while preserving topology. Applied to group-average networks from the two cohorts, diffusion in edge space reduces isolated, noise-like links and clarifies anatomically coherent patterns. This application illustrates the practical value of edge-domain processing for functional connectomics.

\subsection{Data and Preprocessing}
The data analyzed is the Addiction Connectome Preprocessed Initiative (ACPI) distribution of the \emph{Multimodal Treatment of ADHD (MTA)} resting-state fMRI resource released through FCP/INDI. This is a derivative, preprocessed release: raw scans were curated by the project and distributed both as standardized 4D volumes and as parcellated region–wise time series. The parcellated time series in the Automated Anatomical Labeling (AAL) atlas was used for this analysis \cite{tzourio2002automated}. The data is in two groups based on marijuana use status, group that regularly use cannabis (users) and another group that did not regularly use cannabis (non-users). 

Preprocessing in ACPI is provided as a small, explicit factorial set of pipelines that differ only in three choices made after standard motion correction and normalization to MNI space. The first factor is the nonlinear spatial normalization tool (\textit{ANTS} or \textit{FNIRT}); the second is whether high–motion volumes were censored; the third is whether global signal regression was applied. Each subject therefore has up to eight variants, all parcellated to the same AAL ROI set. We selected a single variant \emph{a priori} for all subjects to avoid mixing preprocessing choices across groups. For each subject we transformed the ROI timeseries by removing ROI means, and computing the Pearson correlation between every ROI pair to obtain a symmetric AAL connectivity matrix. Correlations were Fisher $z$–transformed prior to group averaging and back–transformed for reporting and visualization. No additional denoising, filtering, spatial smoothing, or regression was performed beyond the ACPI pipeline settings. This design keeps the preprocessing provenance explicit and comparable across subjects while providing standard parcel–level functional networks for the subsequent edge–domain analyses.

\begin{figure*}[ht]
     \centering
     \begin{subfigure}[t]{0.325\linewidth}
         \centering
         \includegraphics[width=\textwidth]{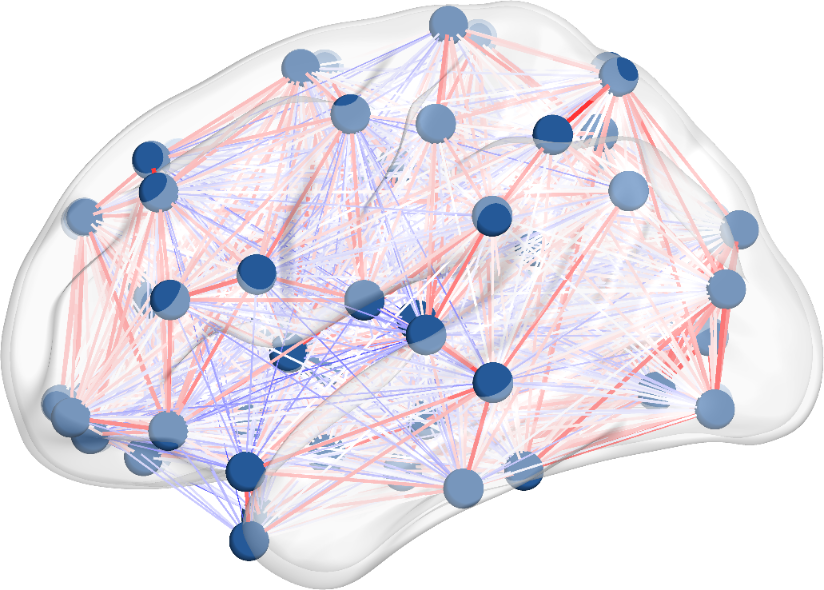}
     \end{subfigure}
     \hfill
     \begin{subfigure}[t]{0.325\linewidth}
         \centering
         \includegraphics[width=0.75\textwidth,height=0.75\textwidth]{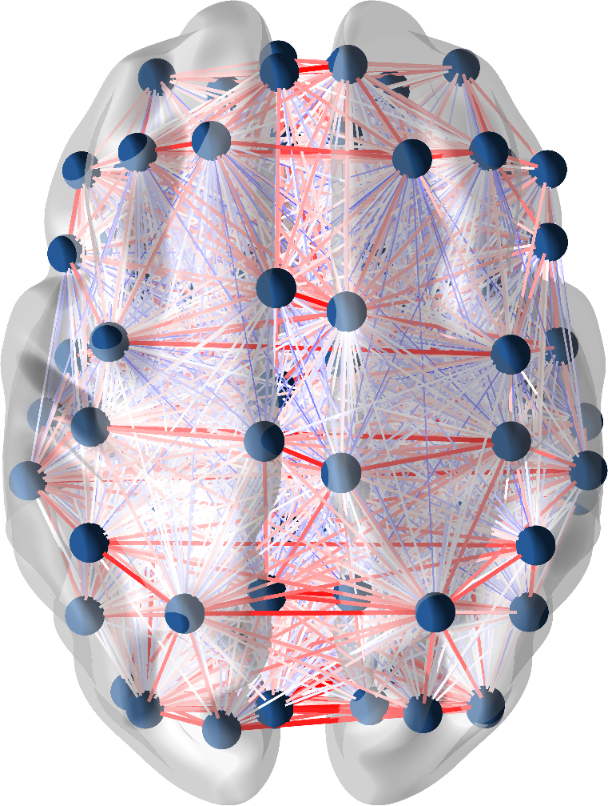}
     \end{subfigure}
     \hfill
     \begin{subfigure}[t]{0.325\linewidth}
         \centering
         \includegraphics[width=\textwidth]{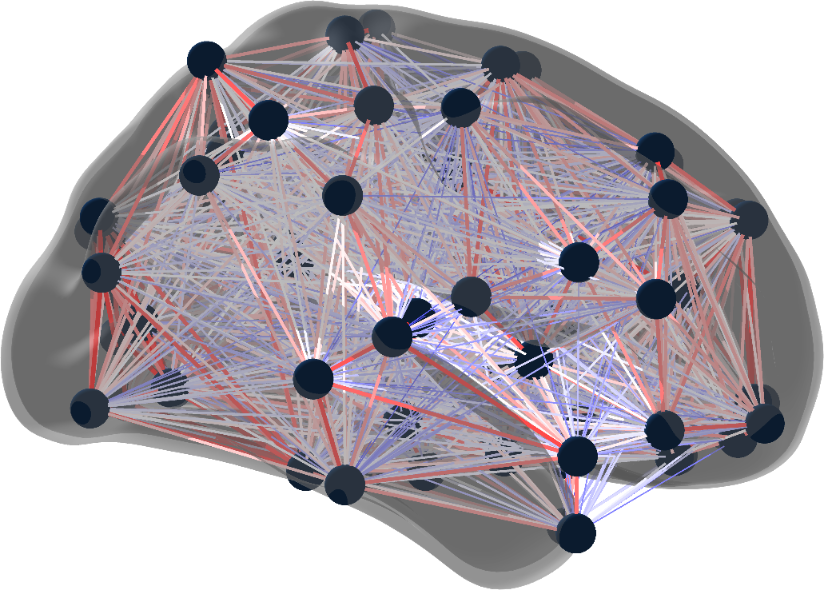}
     \end{subfigure}
     
    \par\vspace{0.5cm}
    
     \centering
     \begin{subfigure}[t]{0.325\linewidth}
         \centering
         \includegraphics[width=\textwidth]{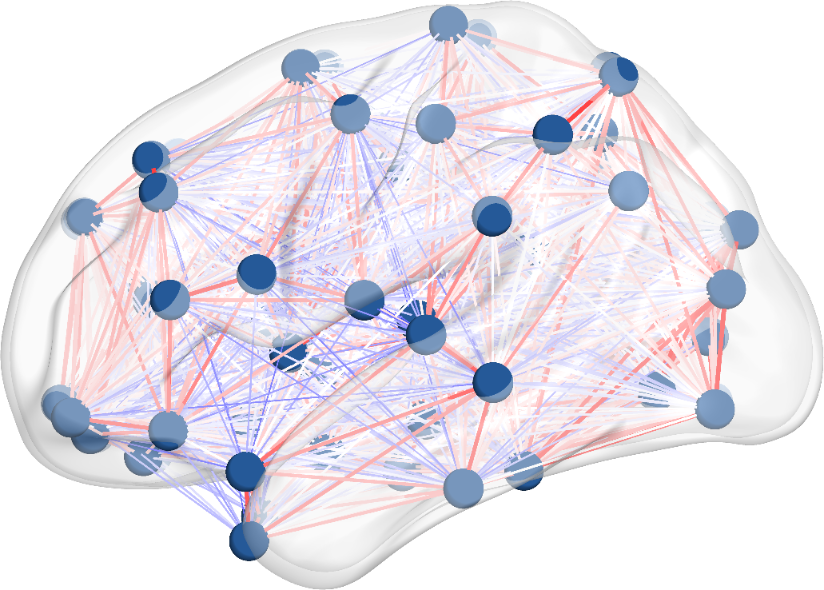}
     \end{subfigure}
     \hfill
     \begin{subfigure}[t]{0.325\linewidth}
         \centering
         \includegraphics[width=0.75\textwidth,height=0.75\textwidth]{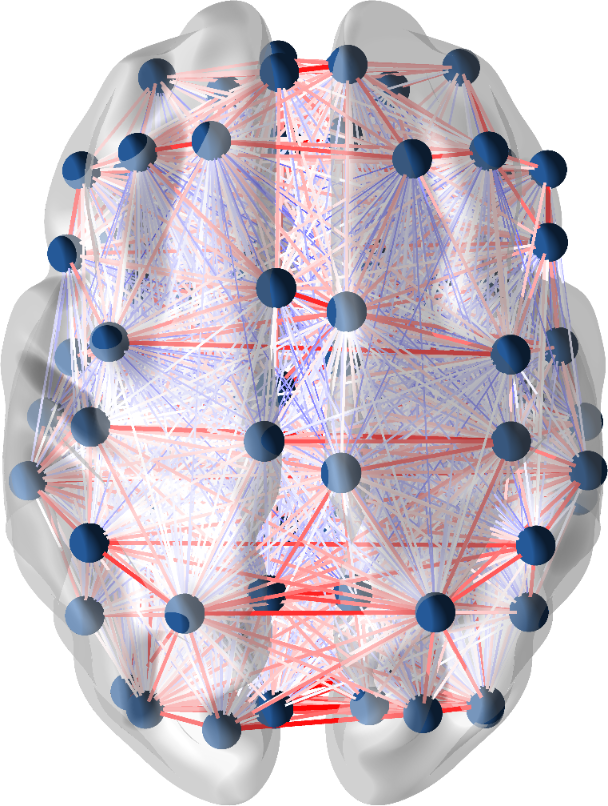}
     \end{subfigure}
     \hfill
     \begin{subfigure}[t]{0.325\linewidth}
         \centering
         \includegraphics[width=\textwidth]{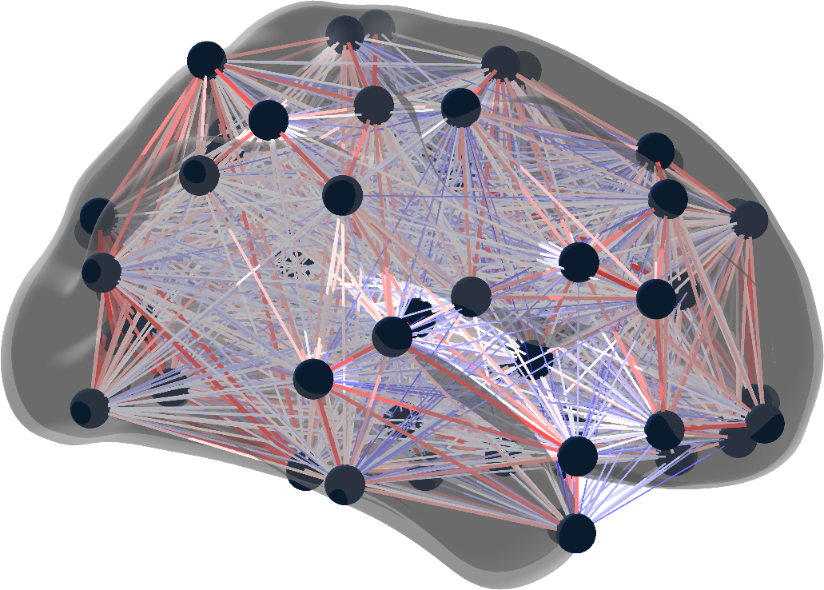}
     \end{subfigure}
      \begin{subfigure}[t]{1\linewidth}
         \centering
         \includegraphics[width=0.75\textwidth]{cb_use1.png}
     \end{subfigure}
    \caption{Top: he functionality connectivity after smoothing along the 1-simplices (edges) with bandwidth $t=0.05$ of the mean connectivity for users. Observe that spurious isolated connections are attenuated (connectivity reduced) while coherent bundles are enhanced. Bottom: The functionality connectivity after smoothing along the 1-simplices (edges) with bandwidth $t=0.05$ of the mean connectivity for non-users. Observe that spurious isolated connections are attenuated (connectivity reduced) while coherent bundles are enhanced.}
    \label{fig:mta-smooth-non-use2}
\end{figure*}

\subsection{Results}
The brain networks in Figure~\ref{fig:mta-raw-non-use} shows the unsmoothed group–mean connectivity for users and non–users respectively. The matrices display the expected within–lobe clustering and interhemispheric homologous links. Observe the dense representation of the connectivities with no visibly discernible pattern. Neuroimaging studies have shown how cannabis affects brain connectivity patterns. Studies examining functional brain connectomes have revealed that both acute and chronic cannabis exposure alter resting-state connectivity, with particular emphasis on networks involved in executive control and self-referential processing\cite{ramaekers2022functional}. THC has been shown to disrupt the DMN, with the posterior cingulate cortex (PCC) being a key region involved in the subjective experience of THC intoxication \cite{wall2019dissociable}. The results in the next section demonstrate that heat kernel smoothing enhance these specific regional patterns without imposing artificial boundaries, which could be very useful in studying such pharmacologically-induced changes.

The functionality connectivity after smoothing along the 1-simplices (edges) are show in Figure~\ref{fig:mta-smooth-non-use2}-top for users and in Figure~\ref{fig:mta-smooth-non-use2}-bottom for the non-users. After edge smoothing, spurious isolated connections are attenuated (reduced connectivity) while coherent bundles are enhanced, yielding clearer interhemispheric and intra–network structure.

Stronger connections can further be realized by increasing the bandwidth as shown in Figure~\ref{fig:mta-smooth-non-use3}.
\begin{figure*}[ht!]
     \centering
     \begin{subfigure}[t]{0.325\linewidth}
         \centering
         \includegraphics[width=\textwidth]{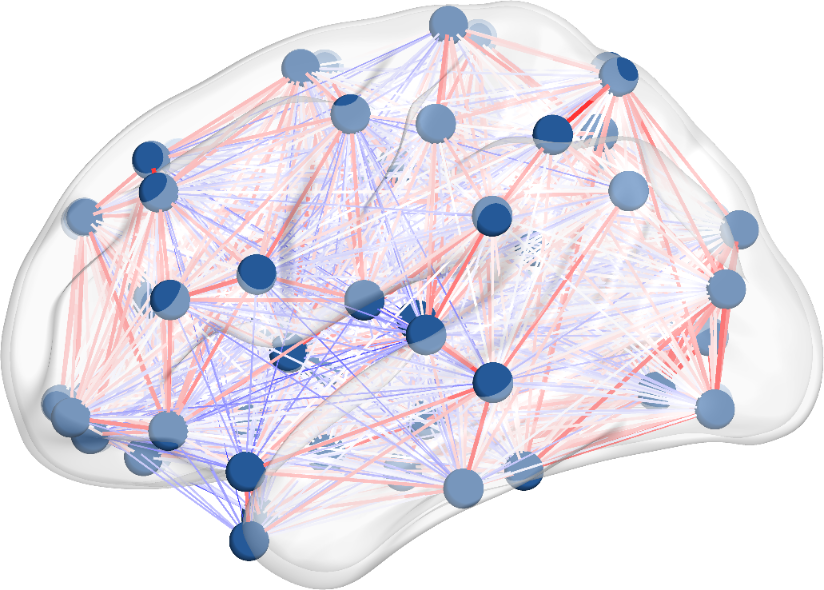}
     \end{subfigure}
     \hfill
     \begin{subfigure}[t]{0.325\linewidth}
         \centering
         \includegraphics[width=0.75\textwidth,height=0.75\textwidth]{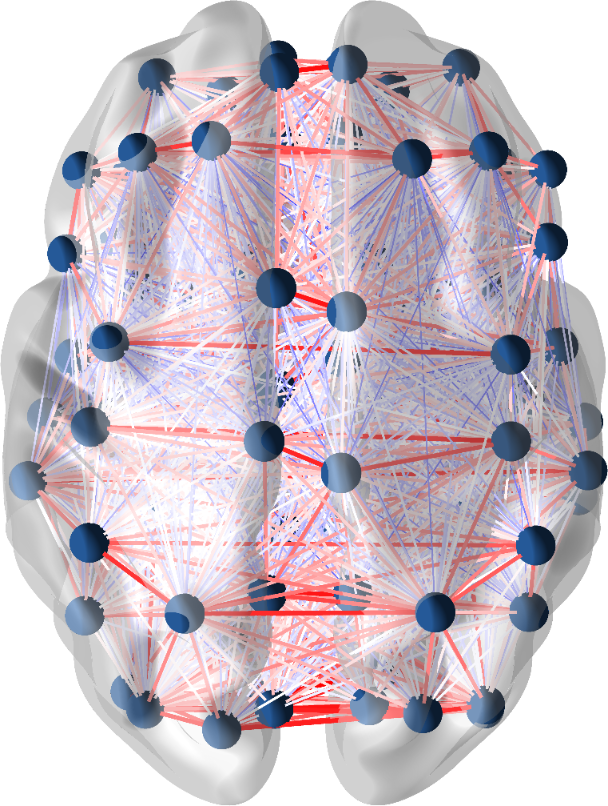}
     \end{subfigure}
     \hfill
     \begin{subfigure}[t]{0.325\linewidth}
         \centering
         \includegraphics[width=\textwidth]{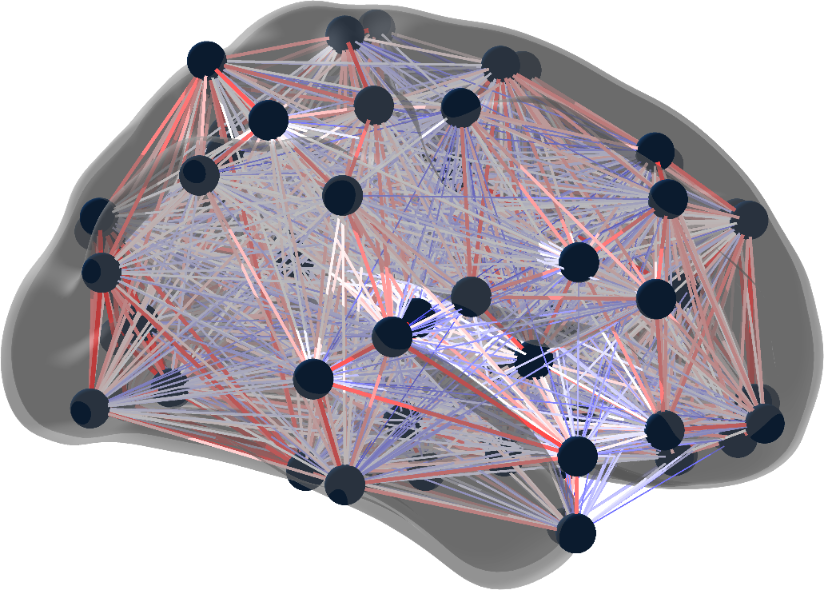}
     \end{subfigure}

    \par\vspace{0.5cm}
    
     \centering
     \begin{subfigure}[t]{0.325\linewidth}
         \centering
         \includegraphics[width=\textwidth]{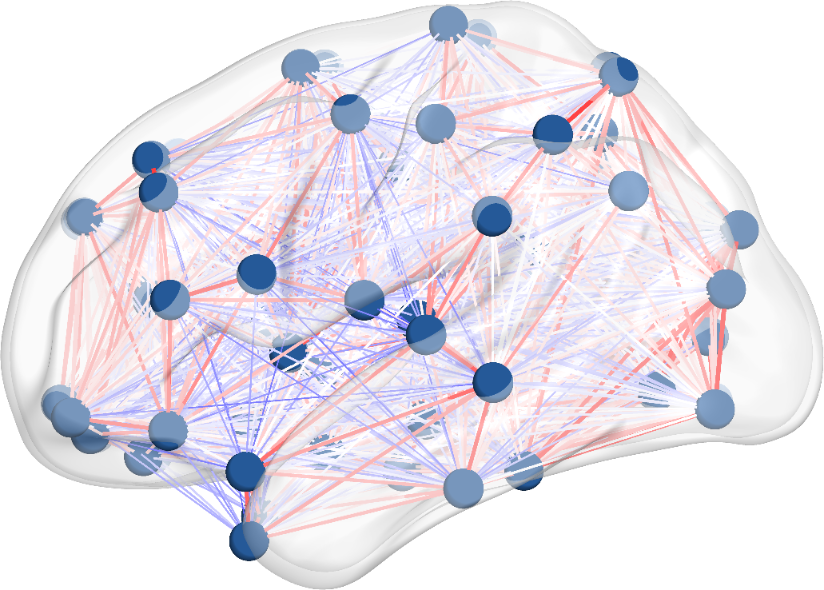}
     \end{subfigure}
     \hfill
     \begin{subfigure}[t]{0.325\linewidth}
         \centering
         \includegraphics[width=0.75\textwidth,height=0.75\textwidth]{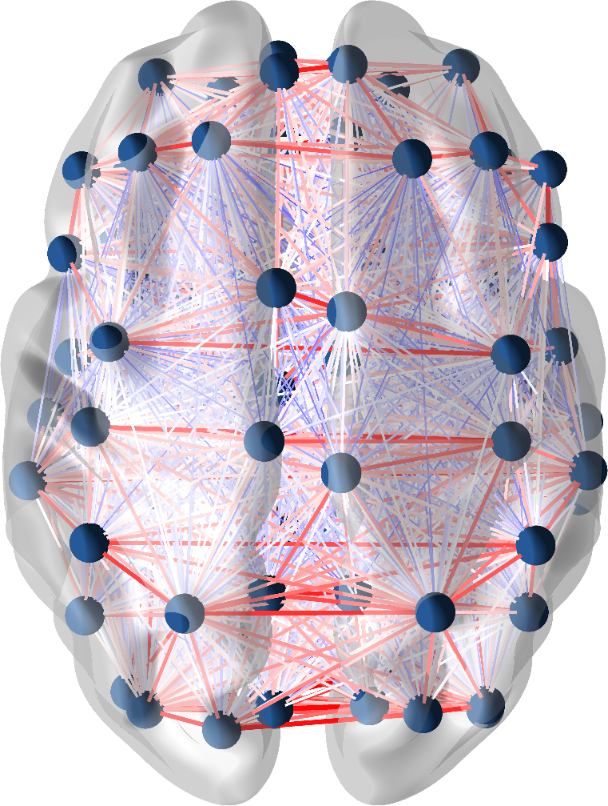}
     \end{subfigure}
     \hfill
     \begin{subfigure}[t]{0.325\linewidth}
         \centering
         \includegraphics[width=\textwidth]{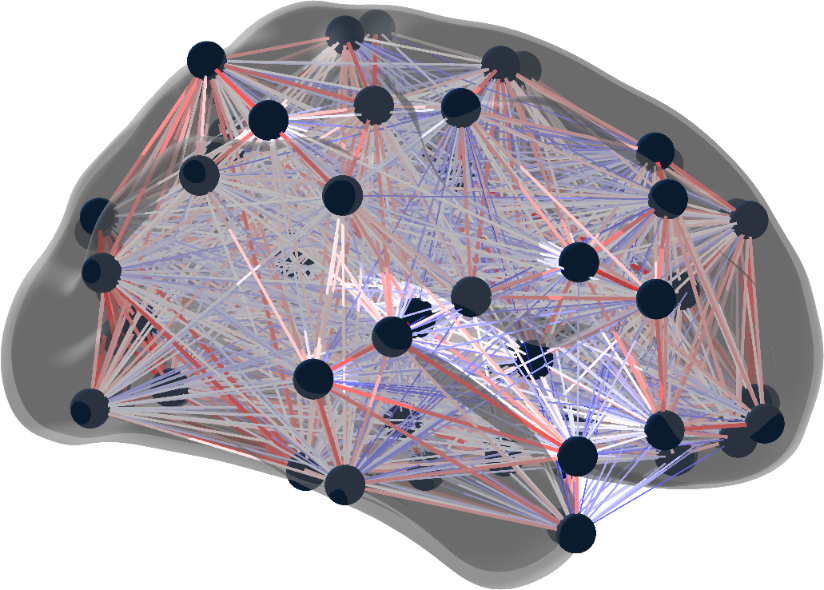}
     \end{subfigure}
      \begin{subfigure}[t]{1\linewidth}
         \centering
         \includegraphics[width=0.75\textwidth]{cb_use1.png}
     \end{subfigure}
    \caption{Top: The functionality connectivity after smoothing along the 1-simplices (edges) with bandwidth $t=0.1$ of the mean connectivity for users. Observe that spurious isolated connections are more attenuated compared to the case where the bandwidth is $t=0.05$. This suppresses isolated noisy connections and sharpens coherent patterns without imposing modular structure or hard thresholds. Bottom: The functionality connectivity after smoothing along the 1-simplices (edges) with bandwidth $t=0.1$ of the mean connectivity for non-users. Observe that spurious isolated connections are more attenuated compared to the case where the bandwidth is $t=0.05$. This suppresses isolated noisy connections and sharpens coherent patterns without imposing modular structure or hard thresholds.}
    \label{fig:mta-smooth-non-use3}
\end{figure*}
The bandwidth parameter selection ($t = 0.05$ and $t = 0.1$) in our analysis corresponds to different scales of network organization. At $t = 0.05$, we observe local edge regularization that preserves fine-grained connectivity patterns while reducing measurement noise. This scale is particularly relevant for detecting subtle alterations in cortico-cortical connections that may reflect compensatory mechanisms in chronic users. At $t = 0.1$, the smoothing reveals larger-scale network architecture, emphasizing major information highways between functional systems.
A direct comparison of the diffused brain networks at bandwidth $t = 0.05$ and bandwidth $t = 0.1$ highlights differences concentrated in fronto–striatal and default–mode edges; these contrasts are uniform in the raw averages but become more spatially organized after diffusion. Glass brain renderings with edges colored by weight visually confirm that smoothing suppresses scattered weak links and emphasizes anatomically plausible tracts. This brain networks illustrations demonstrates the practical application of the proposed kernel smoothing methodology.

\section{Discussion and conclusion} \label{sec:disc}
Simplicial diffusion offers a simple, topology–aware regularizer for functional networks. Applied to the MTA group averages, it suppresses isolated noisy connections and sharpens coherent patterns without imposing modular structure or hard thresholds. Relative to unsmoothed means, the smoothed contrasts between users and non–users are more spatially contiguous and anatomically interpretable. Because the heat kernel operator preserves the simplex-space topology, the procedure does not create spurious $k$-cycles and can be coupled with downstream topological summaries (e.g., loop/non–loop decompositions) for further analysis \cite{anand2021hodge}. Limitations include the choice of bandwidth $t$ and the reliance on a single atlas. In practice $t$ can be selected by minimizing reconstruction error on held–out simplices or by inspecting stability across a small grid; multi–atlas analyses can assess robustness. The present demonstration uses group means, future work would explore subject-level smoothing with appropriate statistical frameworks for group comparisons, potentially using permutation testing on smoothed edge weights or topological features. Integration with dynamic connectivity analyses could reveal how cannabis affects the temporal evolution of network states, with heat kernel smoothing potentially stabilizing time-varying connectivity estimates \cite{dakurah2025robust,dakurah2024subsequence}. Overall, simplex smoothing improves signal-to-noise in functional connectomes with minimal assumptions, producing clearer summaries for visualization and group comparison.

\section{Acknowledgments and Data Availability}
The authors are grateful to the National Institute on Drug Abuse and the Addiction Connectome Preprocessed Initiative (ACPI) for making the functional brain imaging data freely accessible for this work. The neuroimaging data used in this study were obtained from the ADHD-200 Consortium and the ACPI MTA-1 dataset, which are publicly available through the 1000 Functional Connectomes Project at \url{https://fcon_1000.projects.nitrc.org/indi/ACPI/html/acpi_mta_1.html}. The preprocessing pipelines are also publicly available on the 1000 Functional Connectomes Project website.

\bibliographystyle{plain}
\bibliography{references}

@book{bondy1991graph,
  title={Graph Theory with Applications},
  author={Bondy, A.J. and Murty, U.S.R.},
  isbn={9780471363248},
  url={https://books.google.com/books?id=7EWKkgEACAAJ},
  year={1991},
  pages={1--21},
  publisher={Wiley}
}

@book{gross2005graph,
  title={Graph Theory and Its Applications, Second Edition},
  author={Gross, J.L. and Yellen, J.},
  isbn={9781584885054},
  lccn={2005052906},
  series={Textbooks in Mathematics},
  url={https://books.google.com/books?id=-7Q\_POGh-2cC},
  year={2005},
  publisher={Taylor \& Francis}
}

@article{topaz2015topological,
  title={Topological data analysis of biological aggregation models},
  author={Topaz, Chad M and Ziegelmeier, Lori and Halverson, Tom},
  journal={PloS one},
  volume={10},
  number={5},
  pages={e0126383},
  year={2015},
  publisher={Public Library of Science San Francisco, CA USA}
}

@book{hatcher2002algebraic,
  title={Algebraic Topology},
  author={Hatcher, A. and Cambridge University Press and Cornell University. Department of Mathematics},
  isbn={9780521795401},
  lccn={00065166},
  series={Algebraic Topology},
  url={https://books.google.com/books?id=BjKs86kosqgC},
  year={2002},
  publisher={Cambridge University Press}
}

@article{lee2012persistent,
  title={Persistent brain network homology from the perspective of dendrogram},
  author={Lee, Hyekyoung and Kang, Hyejin and Chung, Moo K and Kim, Bung-Nyun and Lee, Dong Soo},
  journal={IEEE transactions on medical imaging},
  volume={31},
  number={12},
  pages={2267--2277},
  year={2012},
  publisher={IEEE}
}

@inproceedings{barbarossa2018learning,
  title={Learning from signals defined over simplicial complexes},
  author={Barbarossa, Sergio and Sardellitti, Stefania and Ceci, Elena},
  booktitle={2018 IEEE Data Science Workshop (DSW)},
  pages={51--55},
  year={2018},
  organization={IEEE}
}

@article{chung2015unified,
  title={Unified heat kernel regression for diffusion, kernel smoothing and wavelets on manifolds and its application to mandible growth modeling in CT images},
  author={Chung, Moo K and Qiu, Anqi and Seo, Seongho and Vorperian, Houri K},
  journal={Medical image analysis},
  volume={22},
  number={1},
  pages={63--76},
  year={2015},
  publisher={Elsevier}
}

@book{sunada2012topological,
  title={Topological Crystallography: With a View Towards Discrete Geometric Analysis},
  author={Sunada, T.},
  isbn={9784431541776},
  lccn={2012953120},
  series={Surveys and Tutorials in the Applied Mathematical Sciences},
  url={https://books.google.com/books?id=6cNEAAAAQBAJ},
  year={2012},
  publisher={Springer Japan}
}

@incollection{adler2010persistent,
  title={Persistent homology for random fields and complexes},
  author={Adler, Robert J and Bobrowski, Omer and Borman, Matthew S and Subag, Eliran and Weinberger, Shmuel},
  booktitle={Borrowing strength: theory powering applications--a Festschrift for Lawrence D. Brown},
  pages={124--143},
  year={2010},
  publisher={Institute of Mathematical Statistics}
}

@article{garlaschelli2009weighted,
  title={The weighted random graph model},
  author={Garlaschelli, Diego},
  journal={New Journal of Physics},
  volume={11},
  number={7},
  pages={073005},
  year={2009},
  publisher={IOP Publishing}
}

@incollection{bollobas1985random,
  title={Random graphs of small order},
  author={Bollob{\'a}s, B{\'e}la and Thomason, Andrew},
  booktitle={North-Holland Mathematics Studies},
  volume={118},
  pages={47--97},
  year={1985},
  publisher={Elsevier}
}

@article{gilbert1959random,
  title={Random graphs},
  author={Gilbert, Edgar N},
  journal={The Annals of Mathematical Statistics},
  volume={30},
  number={4},
  pages={1141--1144},
  year={1959},
  publisher={JSTOR}
}

@article{erdos1960evolution,
  title={On the evolution of random graphs},
  author={Erdos, Paul and R{\'e}nyi, Alfr{\'e}d and others},
  journal={Publ. Math. Inst. Hung. Acad. Sci},
  volume={5},
  number={1},
  pages={17--60},
  year={1960},
  publisher={Citeseer}
}

@article{barabasi1999emergence,
  title={Emergence of scaling in random networks},
  author={Barab{\'a}si, Albert-L{\'a}szl{\'o} and Albert, R{\'e}ka},
  journal={science},
  volume={286},
  number={5439},
  pages={509--512},
  year={1999},
  publisher={American Association for the Advancement of Science}
}

@article{newman2006modularity,
  title={Modularity and community structure in networks},
  author={Newman, Mark EJ},
  journal={Proceedings of the national academy of sciences},
  volume={103},
  number={23},
  pages={8577--8582},
  year={2006},
  publisher={National Acad Sciences}
}

@article{songdechakraiwut2020topological,
  title={Topological Learning for Brain Networks},
  author={Songdechakraiwut, Tananun and Chung, Moo K},
  journal={arXiv preprint arXiv:2012.00675},
  year={2020}
}

@inproceedings{huang2019fast,
  title={Fast polynomial approximation to heat diffusion in manifolds},
  author={Huang, Shih-Gu and Lyu, Ilwoo and Qiu, Anqi and Chung, Moo K},
  booktitle={International Conference on Medical Image Computing and Computer-Assisted Intervention},
  pages={48--56},
  year={2019},
  organization={Springer}
}

@article{RICAUD2019474,
title = {Fourier could be a data scientist: From graph Fourier transform to signal processing on graphs},
journal = {Comptes Rendus Physique},
volume = {20},
number = {5},
pages = {474-488},
year = {2019},
note = {Fourier and the science of today / Fourier et la science d’aujourd’hui},
issn = {1631-0705},
doi = {https://doi.org/10.1016/j.crhy.2019.08.003},
url = {https://www.sciencedirect.com/science/article/pii/S1631070519301094},
author = {Benjamin Ricaud and Pierre Borgnat and Nicolas Tremblay and Paulo Gonçalves and Pierre Vandergheynst},
}

@article{schaub2020random,
  title={Random walks on simplicial complexes and the normalized hodge 1-laplacian},
  author={Schaub, Michael T and Benson, Austin R and Horn, Paul and Lippner, Gabor and Jadbabaie, Ali},
  journal={SIAM Review},
  volume={62},
  number={2},
  pages={353--391},
  year={2020},
  publisher={SIAM}
}

@inproceedings{muhammad2006control,
  title={Control using higher order Laplacians in network topologies},
  author={Muhammad, Abubakr and Egerstedt, Magnus},
  booktitle={Proc. of 17th International Symposium on Mathematical Theory of Networks and Systems},
  pages={1024--1038},
  year={2006},
  organization={Citeseer}
}

@article{hagmann2008mapping,
  title={Mapping the structural core of human cerebral cortex},
  author={Hagmann, Patric and Cammoun, Leila and Gigandet, Xavier and Meuli, Reto and Honey, Christopher J and Wedeen, Van J and Sporns, Olaf},
  journal={PLoS biology},
  volume={6},
  number={7},
  pages={e159},
  year={2008},
  publisher={Public Library of Science San Francisco, USA}
}

@article{macdonald2000automated,
  title={Automated 3-D extraction of inner and outer surfaces of cerebral cortex from MRI},
  author={MacDonald, David and Kabani, Noor and Avis, David and Evans, Alan C},
  journal={NeuroImage},
  volume={12},
  number={3},
  pages={340--356},
  year={2000},
  publisher={Elsevier}
}

@article{ortega2018graph,
  title={Graph signal processing: Overview, challenges, and applications},
  author={Ortega, Antonio and Frossard, Pascal and Kova{\v{c}}evi{\'c}, Jelena and Moura, Jos{\'e} MF and Vandergheynst, Pierre},
  journal={Proceedings of the IEEE},
  volume={106},
  number={5},
  pages={808--828},
  year={2018},
  publisher={IEEE}
}

@article{huang2018graph,
  title={A graph signal processing perspective on functional brain imaging},
  author={Huang, Weiyu and Bolton, Thomas AW and Medaglia, John D and Bassett, Danielle S and Ribeiro, Alejandro and Van De Ville, Dimitri},
  journal={Proceedings of the IEEE},
  volume={106},
  number={5},
  pages={868--885},
  year={2018},
  publisher={IEEE}
}

@article{hu2016matched,
  title={Matched signal detection on graphs: Theory and application to brain imaging data classification},
  author={Hu, Chenhui and Sepulcre, Jorge and Johnson, Keith A and Fakhri, Georges E and Lu, Yue M and Li, Quanzheng},
  journal={NeuroImage},
  volume={125},
  pages={587--600},
  year={2016},
  publisher={Elsevier}
}

@article{ganmor2011sparse,
  title={Sparse low-order interaction network underlies a highly correlated and learnable neural population code},
  author={Ganmor, Elad and Segev, Ronen and Schneidman, Elad},
  journal={Proceedings of the National Academy of sciences},
  volume={108},
  number={23},
  pages={9679--9684},
  year={2011},
  publisher={National Acad Sciences}
}

@article{yu2011higher,
  title={Higher-order interactions characterized in cortical activity},
  author={Yu, Shan and Yang, Hongdian and Nakahara, Hiroyuki and Santos, Gustavo S and Nikoli{\'c}, Danko and Plenz, Dietmar},
  journal={Journal of neuroscience},
  volume={31},
  number={48},
  pages={17514--17526},
  year={2011},
  publisher={Soc Neuroscience}
}

@article{ohiorhenuan2010sparse,
  title={Sparse coding and high-order correlations in fine-scale cortical networks},
  author={Ohiorhenuan, Ifije E and Mechler, Ferenc and Purpura, Keith P and Schmid, Anita M and Hu, Qin and Victor, Jonathan D},
  journal={Nature},
  volume={466},
  number={7306},
  pages={617--621},
  year={2010},
  publisher={Nature Publishing Group}
}

@book{newman2018networks,
  title={Networks},
  author={Newman, Mark},
  year={2018},
  publisher={Oxford university press}
}

@article{bullmore2009complex,
  title={Complex brain networks: graph theoretical analysis of structural and functional systems},
  author={Bullmore, Ed and Sporns, Olaf},
  journal={Nature reviews neuroscience},
  volume={10},
  number={3},
  pages={186--198},
  year={2009},
  publisher={Nature Publishing Group}
}

@article{giusti2016two,
  title={Two’s company, three (or more) is a simplex},
  author={Giusti, Chad and Ghrist, Robert and Bassett, Danielle S},
  journal={Journal of computational neuroscience},
  volume={41},
  number={1},
  pages={1--14},
  year={2016},
  publisher={Springer}
}

@article{andjelkovic2020topology,
  title={The topology of higher-order complexes associated with brain hubs in human connectomes},
  author={Andjelkovi{\'c}, Miroslav and Tadi{\'c}, Bosiljka and Melnik, Roderick},
  journal={Scientific reports},
  volume={10},
  number={1},
  pages={1--10},
  year={2020},
  publisher={Nature Publishing Group}
}

@article{emrani2014persistent,
  title={Persistent homology of delay embeddings and its application to wheeze detection},
  author={Emrani, Saba and Gentimis, Thanos and Krim, Hamid},
  journal={IEEE Signal Processing Letters},
  volume={21},
  number={4},
  pages={459--463},
  year={2014},
  publisher={IEEE}
}

@article{emrani2015novel,
  title={A novel framework for pulse pressure wave analysis using persistent homology},
  author={Emrani, Saba and Saponas, T Scott and Morris, Dan and Krim, Hamid},
  journal={IEEE Signal Processing Letters},
  volume={22},
  number={11},
  pages={1879--1883},
  year={2015},
  publisher={IEEE}
}

@book{rosenberg1997laplacian,
  title={The Laplacian on a Riemannian manifold: an introduction to analysis on manifolds},
  author={Rosenberg, Steven and Steven, Rosenberg},
  number={31},
  year={1997},
  publisher={Cambridge University Press}
}

@article{morgan1974embedding,
  title={Embedding metric spaces in Euclidean space},
  author={Morgan, CL},
  journal={Journal of Geometry},
  volume={5},
  number={1},
  pages={101--107},
  year={1974},
  publisher={Springer}
}

@inproceedings{schaub2018flow,
  title={Flow smoothing and denoising: Graph signal processing in the edge-space},
  author={Schaub, Michael T and Segarra, Santiago},
  booktitle={2018 IEEE Global Conference on Signal and Information Processing (GlobalSIP)},
  pages={735--739},
  year={2018},
  organization={IEEE}
}

@article{barbarossa2020topological,
  title={Topological signal processing over simplicial complexes},
  author={Barbarossa, Sergio and Sardellitti, Stefania},
  journal={IEEE Transactions on Signal Processing},
  volume={68},
  pages={2992--3007},
  year={2020},
  publisher={IEEE}
}

@article{edelsbrunner2008persistent,
  title={Persistent homology-a survey},
  author={Edelsbrunner, Herbert and Harer, John and others},
  journal={Contemporary mathematics},
  volume={453},
  pages={257--282},
  year={2008},
  publisher={Providence, RI: American Mathematical Society}
}

@inproceedings{chung2019statistical,
  title={Statistical inference on the number of cycles in brain networks},
  author={Chung, Moo K and Huang, Shih-Gu and Gritsenko, Andrey and Shen, Li and Lee, Hyekyoung},
  booktitle={2019 IEEE 16th International Symposium on Biomedical Imaging (ISBI 2019)},
  pages={113--116},
  year={2019},
  organization={IEEE}
}

@article{park2013structural,
  title={Structural and functional brain networks: from connections to cognition},
  author={Park, Hae-Jeong and Friston, Karl},
  journal={Science},
  volume={342},
  number={6158},
  year={2013},
  publisher={American Association for the Advancement of Science}
}

@article{tarjan1972depth,
  title={Depth-first search and linear graph algorithms},
  author={Tarjan, Robert},
  journal={SIAM journal on computing},
  volume={1},
  number={2},
  pages={146--160},
  year={1972},
  publisher={SIAM}
}

@article{anand2021hodge,
  title={Hodge-Laplacian of Brain Networks and Its Application to Modeling Cycles},
  author={Anand, D Vijay and Chung, Moo K},
  journal={arXiv preprint arXiv:2110.14599},
  year={2021}
}

@article{zomorodian2005computing,
  title={Computing persistent homology},
  author={Zomorodian, Afra and Carlsson, Gunnar},
  journal={Discrete \& Computational Geometry},
  volume={33},
  number={2},
  pages={249--274},
  year={2005},
  publisher={Springer}
}

@inproceedings{busaryev2012annotating,
  title={Annotating simplices with a homology basis and its applications},
  author={Busaryev, Oleksiy and Cabello, Sergio and Chen, Chao and Dey, Tamal K and Wang, Yusu},
  booktitle={Scandinavian workshop on algorithm theory},
  pages={189--200},
  year={2012},
  organization={Springer}
}

@article{dakurahregistration,
  title={Registration and Joint Identification of Cycles in Brain Networks},
  author={Dakurah, Sixtus and Chung, Moo K},
  journal={Preprint},
  year={2023}
}

@inproceedings{dakurah2022modelling,
  title={Modelling cycles in brain networks with the Hodge Laplacian},
  author={Dakurah, Sixtus and Anand, D Vijay and Chen, Zijian and Chung, Moo K},
  booktitle={International Conference on Medical Image Computing and Computer-Assisted Intervention},
  pages={326--335},
  year={2022},
  organization={Springer}
}

@article{dakurah2025maxtda,
  title={MaxTDA: Robust Statistical Inference for Maximal Persistence in Topological Data Analysis},
  author={Dakurah, Sixtus and Cisewski-Kehe, Jessi},
  journal={arXiv preprint arXiv:2504.03897},
  year={2025}
}

@book{dakurah2025robust,
  title={Robust Statistical Methods for Topological Data Analysis of Time Series},
  author={Dakurah, Sixtus},
  year={2025},
  publisher={The University of Wisconsin-Madison}
}

@article{dakurah2024subsequence,
  title={A subsequence approach to topological data analysis for irregularly-spaced time series},
  author={Dakurah, Sixtus and Cisewski-Kehe, Jessi},
  journal={arXiv preprint arXiv:2410.13723},
  year={2024}
}

@article{tzourio2002automated,
  title={Automated anatomical labeling of activations in SPM using a macroscopic anatomical parcellation of the MNI MRI single-subject brain},
  author={Tzourio-Mazoyer, Nathalie and Landeau, Brigitte and Papathanassiou, Dimitri and Crivello, Fabrice and Etard, Olivier and Delcroix, Nicolas and Mazoyer, Bernard and Joliot, Marc},
  journal={Neuroimage},
  volume={15},
  pages={273--289},
  year={2002},
  publisher={Elsevier}
}

@article{ramaekers2022functional,
  title={Functional brain connectomes reflect acute and chronic cannabis use},
  author={Ramaekers, Johannes G and Mason, Natasha L and T{\"o}nnes, Stefan W and Theunissen, Eef L and Amico, Enrico},
  journal={Scientific reports},
  volume={12},
  number={1},
  pages={2449},
  year={2022},
  publisher={Nature Publishing Group UK London}
}

@article{wall2019dissociable,
  title={Dissociable effects of cannabis with and without cannabidiol on the human brain’s resting-state functional connectivity},
  author={Wall, Matthew B and Pope, Rebecca and Freeman, Tom P and Kowalczyk, Oliwia S and Demetriou, Lysia and Mokrysz, Claire and Hindocha, Chandni and Lawn, Will and Bloomfield, Michael AP and Freeman, Abigail M and others},
  journal={Journal of Psychopharmacology},
  volume={33},
  number={7},
  pages={822--830},
  year={2019},
  publisher={Sage Publications Sage UK: London, England}
}

\end{document}